\documentclass[journal]{IEEEtran}

\ifCLASSINFOpdf
\else
\fi
\usepackage{cite}
\usepackage{subcaption}
\usepackage{amsmath,amssymb,amsfonts}
\usepackage{amsthm}
\usepackage{fancyhdr}
\usepackage{algorithmic}
\usepackage{authblk}
\usepackage[ruled,linesnumbered,lined,boxed,commentsnumbered]{algorithm2e}
\usepackage{graphicx}
\usepackage{hyperref}
\usepackage{textcomp}
\usepackage{enumitem}
\usepackage{caption}
\usepackage{pifont}
\usepackage{commath}
\usepackage{xcolor}
\usepackage{upgreek}
\usepackage{multirow}
\usepackage{mdframed}
\usepackage{booktabs} 
\usepackage{array}    
\usepackage{siunitx}  

\newtheorem{condition}{Condition}
\newtheorem{remark}{Remark}

\newtheorem{definition}{Definition}
\newtheorem{theorem}{Theorem}
\newtheorem{lemma}{Lemma}
\newtheorem{challenge}{Challenge}
\usepackage{bm}
\usepackage[left=0.63in,top=0.7in,right=0.63in,bottom=0.95in]{geometry}

\usepackage{afterpage}

\begin{document}
\title{\fontsize{22}{25}\selectfont PrivTuner with Homomorphic Encryption and LoRA: \\ A P3EFT Scheme for Privacy-Preserving
Parameter-Efficient Fine-Tuning of AI Foundation Models}

\author{Yang~Li, Wenhan~Yu, Jun~Zhao
\thanks{The authors are all with the College of Computing and Data Science, Nanyang Technological University (NTU), Singapore. Yang Li is also with ERI@N, Interdisciplinary Graduate Programme, NTU, Singapore. Emails: yang048@e.ntu.edu.sg, wenhan002@e.ntu.edu.sg, JunZhao@ntu.edu.sg.\\ Corresponding author: Jun Zhao
} 
}

\maketitle
\begin{abstract}
AI foundation models have recently demonstrated impressive capabilities across a wide range of tasks. Fine-tuning (FT) is a method of customizing a pre-trained AI foundation model by further training it on a smaller, targeted dataset. In this paper, we initiate the study of the Privacy-Preserving
Parameter-Efficient FT (P3EFT) framework, which can be viewed as the intersection of Parameter-Efficient FT (PEFT)
and Privacy-Preserving FT (PPFT). PEFT modifies only a small subset of the model's parameters to achieve FT (i.e., adapting a pre-trained model to a specific dataset), while PPFT uses privacy-preserving technologies to protect the confidentiality of the model during the FT process. There have been many studies on PEFT or PPFT, but very few on their fusion, which motivates our work on P3EFT to achieve both parameter efficiency and model privacy.    
To exemplify our P3EFT, we present the \textit{PrivTuner} scheme, which incorporates Fully Homomorphic Encryption (FHE) enabled privacy protection into LoRA (short for ``Low-Rank Adapter''), a popular PEFT solution published in ICLR 2021~\cite{hu2021lora}.
Intuitively speaking, PrivTuner allows the model owner and the external data owners to collaboratively implement PEFT with encrypted data. After describing PrivTuner in detail, we further investigate its energy consumption and privacy protection. Then, we consider a PrivTuner system over wireless communications and formulate a joint optimization problem to adaptively minimize energy while maximizing privacy protection, with the optimization variables including FDMA bandwidth allocation, wireless transmission power, computational resource allocation, and privacy protection.
A resource allocation algorithm is devised to solve the problem.
Experiments demonstrate that our algorithm can significantly reduce energy consumption while adapting to different privacy requirements.

\end{abstract}


\IEEEpeerreviewmaketitle
\section{Introduction} \label{secIntroduction}

\textbf{AI Foundation Models.} In the current vibrant era of Artificial Intelligence (AI), large foundation models like BERT~\cite{devlin2018bert}, CLIP~\cite{radford2019language}, and \mbox{GPT-3}~\cite{brown2020language} 
have ushered in a significant revolution, moving beyond conventional machine-learning techniques.
These models often have a large number of parameters and sophisticated structures, which can discern nuanced context features, even rivalling human proficiency.
Through pre-training on massive datasets, foundation models learn a wide range of general features and patterns, enabling their effective adaptation to different downstream tasks via a transfer learning method known as fine-tuning~\cite{ouyang2022training}. For instance, vanilla BERT can be easily fine-tuned to tasks including machine translation~\cite{dai2022bertology}, sentiment analysis~\cite{hoang2019aspect}, abstract summary~\cite{wang2019text} and question answering~\cite{wang2019multi}, achieving satisfactory performance.
This fine-tuning process is pivotal, as it customizes those foundation models for the specific needs of applications, providing better personalized services to users.

\textbf{Fine Tuning (FT).} Tuning foundation models is challenging, as discussed below. First, with the rise in concern about data privacy, entities are increasingly reluctant to let their data leave their local devices, which makes centralized training difficult to implement~\cite{zhou2023resource}. 
Additionally, tuning a foundation model with a large number of parameters is usually computationally intensive, thus causing high training consumption. If the entity desires local fine-tuning, this can pose a large burden on it, especially in computational resource-limited scenarios.

\textbf{Parameter-Efficient FT: PEFT.} 
To address the above challenges, various PEFT methods have been proposed, including Low-Rank Adapter (LoRA)~\cite{hu2021lora}, prefix-tuning~\cite{li2021prefix}, prompt tuning~\cite{qin2021learning}, BitFit~\cite{zaken2021bitfit}, and more. PEFT allows entities to update only a small portion of the entire model's parameters, thus reducing the computational burden in fine-tuning progress. However, most existing PEFT methods permit external entities access to the full foundation model, neglecting concerns regarding model privacy. Institutes or companies often invest huge amounts of financial support in training and maintaining a well-performing foundation model (e.g., OpenAI spends approximately \$700,000 daily to operate ChatGPT~\cite{srivastawa2023exploring}). As a result, they are often unwilling to share the model with external entities for various reasons, such as intellectual property, profitability, concerns about abuse, etc. 
Consequently, it becomes imperative to explore the development of a {model- and data-privacy-friendly} fine-tuning framework.

\textbf{Privacy-Preserving FT: PPFT.} 
PPFT is a subset of Privacy-Preserving Machine Learning (PPML) and proposed in response to privacy concerns in fine-tuning, using various Privacy-Preserving Technologies (PPTech)\footnote{Terms related to PPTech include Privacy-Enhancing Technologies (PET), which has also been widely used. We adopt the phrase ``privacy-preserving'' instead of ``privacy-enhancing'' to avoid the abbreviation of ``Privacy-Enhancing FT'' being PEFT, since PEFT is reserved for ``Parameter-Efficient FT''.}, e.g., Federated Learning (FL)~\cite{chen2023federated,yu2023federated}, Differential Privacy (DP)~\cite{yu2021differentially}, secure Multi-Party Computing (MPC)~\cite{ramachandran2021s++} and Homomorphic Encryption (HE)~\cite{al2020privft,lee2022privacy}.
However, most existing PPFT methods achieve privacy at the expense of communication or computing resources. For instance, FL allows each entity to conduct local fine-tuning thus ensuring data privacy, but requires additional communication overhead. Therefore, it is also worth investigating how to {mitigate the consumption} of existing PPFT methods, e.g., latency and energy.


Based on the above, we compare the characteristics of the above two approaches:
\begin{itemize}
    \item Parameter-Efficient FT (PEFT) effectively reduces the computation burden, but existing PEFT methods often overlook the privacy issue.
    \item Privacy-Preserving FT (PPFT) could protect data and model privacy, but it often introduces additional communication and computation overheads.
\end{itemize}

We recognize that PEFT and PPFT do not necessarily have to function as exclusive mechanisms. Instead, their integration can provide mutual support, enhancing privacy under limited computational budgets. This motivates our research explained below.
     



\textbf{Privacy-Preserving Parameter-Efficient FT: P3EFT studied in this paper.}  We initiate the study of  \textit{Privacy-Preserving Parameter-Efficient Fine Tuning (P3EFT)}. P3EFT can be understood as the intersection of PEFT and PPFT. 
To demonstrate the feasibility of P3EFT, we propose a specific instantiation, \textit{PrivTuner}, which leverages Fully Homomorphic Encryption (FHE) and LoRA for privacy-preserving and computation-efficient fine-tuning. By integrating FHE, PrivTuner enables computations to be performed directly on encrypted data, preventing data exposure while maintaining the efficiency of LoRA-based adaptation. This design allows external devices to fine-tune models collaboratively with minimal trust assumptions, making it suitable for scenarios such as federated learning, cloud-based AI services, and edge intelligence. We will detail the system model of PrivTuner in Section~\ref{sec:system}.


\textbf{Resource Allocation in PrivTuner.} Resource Allocation (RA) plays a pivotal role in the wireless communication network where resources (e.g., transmission power and bandwidth) are often limited. In our paper, we consider PrivTuner to operate in a wireless communication network comprising one server and multiple clients, and RA could effectively manage different resources to improve the overall network performance. 
By incorporating RA, we can formulate a customized optimization problem for the PrivTuner network that aligns with specific resource budgets and performance demands (e.g., energy consumption, latency and privacy) and then design an efficient algorithm to solve it. The problem formulation and RA algorithm design are discussed in Sections~\ref{sec:problem} and \ref{sec:algorithm}.

\textbf{Contributions.} Our main contributions are as follows:
\begin{itemize}
    \item 
    We introduce the \textit{P3EFT} framework, which bridges the gap between PEFT and PPFT by integrating their advantages to enable both efficient and privacy-preserving fine-tuning. To validate this framework, we propose \textit{PrivTuner} as a concrete instance of P3EFT. PrivTuner combines LoRA with FHE to provide secure and efficient collaborative fine-tuning between the model owner and external devices, reducing the computational burden while ensuring data and model privacy.
    
    \item We investigate the time and energy consumption models within PrivTuner from both the server and device perspectives. We also provide a metric to evaluate the overall privacy protection level of PrivTuner.
    A joint problem is formulated by optimizing computation capacities, wireless communication resources and FHE settings to minimize energy consumption while maximizing privacy protection.
    \item  A resource allocation algorithm employing both branch and bound (B\&B) and fractional programming techniques is devised to solve the formulated problem effectively. Time complexity, solution quality and convergence analysis are also provided. Experimental results validate the effectiveness and superiority of our algorithm.
    
\end{itemize}

The rest of the paper is organized as follows. Section \ref{sec:review} presents preliminaries and Section \ref{secRelated} reviews related work. 
The PrivTuner scheme and its consumption models are introduced in Section \ref{sec:system}. Section~\ref{secConsumption} models the consumption and privacy protection within PrivTuner. Section \ref{sec:problem} presents the problem formulation in detail. Section~\ref{sec:algorithm} introduces our RA algorithm with time, solution and convergence analysis. Section \ref{sec:experimental results} gives the fitting performance, experimental setting and results. Section \ref{sec:7} concludes the paper.

\section{Preliminaries}\label{sec:review}

This section presents preliminaries related to our research.  

\subsection{Low-Rank Adapters (LoRA) for AI Foundation Models}\label{subsec:lora}
The motivation of LoRA~\cite{hu2021lora} is that the updates to pre-trained model weights are likely to have a low intrinsic rank during adaption. Hence, for a pre-trained weight matrix $W_0 \in \mathbb{R}^{d\times k}$, the update is constrained by a low-rank decomposition: $W_0+\Delta W=W_0+A^1A^2$, where $A^1\in\mathbb{R}^{d\times r},A^2 \in \mathbb{R}^{r \times k}$, and the rank $r \ll \min\{d,k\}$. During training, $W_0$ is frozen and not updated, while $A^1$ and $A^2$ contain trainable parameters. When given an input $x$ to the model, the forward pass yields:
\begin{align}
    y = W_0x+ \Delta W x = W_0x+A^1A^2x.
\end{align}
LoRA initializes $A^2$ as a random Gaussian matrix and $A^1$ as a zero matrix. 

\subsection{Fully Homomorphic Encryption (FHE)} \label{secfhe}
Fully Homomorphic Encryption stands as an advanced cryptographic technique, pioneered by~\cite{gentry2009fully}, 
supporting computations to be executed on encrypted data directly. The implications of FHE are profound for security across diverse applications, notably in privacy-sensitive scenarios.
FHE has evolved with several typical variations, including Brakerski/Fan--Vercauteren
(BFV)~\cite{brakerski2012fully} scheme, Brakerski--Gentry--Vaikuntanathan (BGV)~\cite{brakerski2014leveled}, and Cheon--Kim--Kim--Song (CKKS)~\cite{cheon2017homomorphic}. CKKS has gained popularity due to its extraordinary performance in handling real numbers, and we select it as the FHE scheme in this paper. 
An overview of the algorithms in CKKS is as follows:
\begin{itemize}
    \item $\mathrm{KeyGen}(\lambda,q)\rightarrow(\mathrm{pk},\mathrm{sk})$: On input a polynomial degree $\lambda$ and a coefficient modulus $q$, randomly outputs a public key $(\mathrm{pk})$ and a secret key $(\mathrm{sk})$. 
    \item $\mathrm{Enc}(\mathrm{pk},m)\rightarrow\widetilde{c}$: On public key $\mathrm{pk}$ and a plaintext $m$, perform encryption to obtain a ciphertext $\widetilde{c}$.
    \item $\mathrm{Dec}(\mathrm{sk},\widetilde{c})\rightarrow{m}$: On secret key $\mathrm{sk}$ and a ciphertext $\widetilde{c}$, perform decryption to obtain the plaintext message $m$.
    \item $\mathrm{Eval}(\mathrm{pk},\widetilde{c}_1,\widetilde{c}_2,\mathrm{func})\rightarrow\widetilde{c}$: On public key $\mathrm{pk}$, two ciphertexts $\widetilde{c}_1,\widetilde{c}_2$ encrypted from $m_1,m_2$, and a linear function $\mathrm{func}$, outputs a new ciphertext $\widetilde{c}$ such that it is the same as $\mathrm{Enc}(\mathrm{pk},\mathrm{func}(m_1,m_2))$.
\end{itemize}
A more detailed discussion of CKKS can be found in~\cite{cheon2017homomorphic}.
Specifically, CKKS is designed for approximate arithmetic on real and complex numbers, distinguishing it from integer-based schemes (e.g., BFV, BGV). In addition, CKKS effectively manages the noise growth through customized scaling and rescaling techniques, which could prevent the overflow issue. The above characteristics make CKKS well-suited for privacy-preserving deep learning tasks. Deep learning models often operate on real-valued numbers (e.g., floating-point numbers) rather than exact integers. These models are also robust to small numerical errors or approximations due to the reliance on optimization techniques and probabilistic learning. Consequently, CKKS is widely used in privacy-preserving applications across various AI models, including CNNs~\cite{mihara2020neural,al2020privft,lee2022privacy}, RNNs~\cite{bakshi2020cryptornn,podschwadt2021non}, and Transformers~\cite{zhang2024secure,rovida2024transformer}.

\begin{figure}[t] 
\centering
\includegraphics[width=1\linewidth]{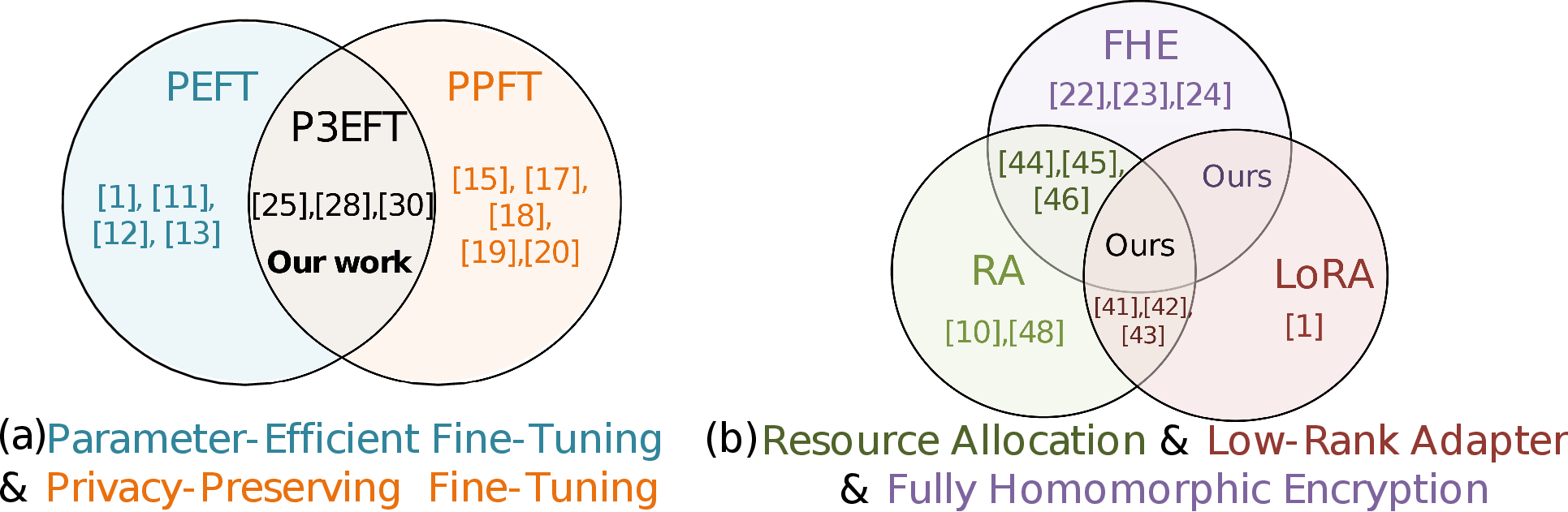}
\caption{The relationships among the concepts.}
\vspace{-0.05cm}
\label{fig:relationships}
\end{figure}

\section{Related Work} \label{secRelated}

We discuss additional work related to our research. Some related studies have already been mentioned in Section~\ref{secIntroduction}.

To guide the discussion of related work, we plot Figures \ref{fig:relationships}(a) and (b) at different granularities.
Figure \ref{fig:relationships}(a) presents a Venn diagram of general technologies, where we examine PEFT, PPFT, and their intersection: P3EFT. Note that PEFT, PPFT, and P3EFT are all high-level notions of the technologies, and do not consider the specific implementations.
Figure \ref{fig:relationships}(b) is about specific implementations of the technologies. In particular, we instance PEFT as LoRA, and PPFT as being enabled by FHE. In addition, we also plot ``Resource Allocation (RA)'' as a circle in the Venn diagram of Figure \ref{fig:relationships}(b), since the technical analysis of our paper is about RA.

\subsection{Discussing Related Studies based on Figure \ref{fig:relationships}(a)} 

The related studies in Fig.~\ref{fig:relationships}(a) include those related to 1) P3EFT, 2) PEFT and 3) PPFT.

After proposing the term P3EFT, we found a recent unpublished work~\cite{zmushko2023privacy} which also coins P3EFT independently. The study~\cite{zmushko2023privacy} incorporates \textit{split learning}~\cite{abuadbba2020can} into PEFT to achieve only heuristic privacy without even defining what protection is being offered (in fact, \textit{split learning} is shown to be inherently
insecure~\cite{pasquini2021unleashing}). In addition, Bu~\emph{et~al.}~\cite{bu2022differentially} incorporate the formal notion of differential privacy (DP) into PEFT to achieve P3EFT, but the name P3EFT is not mentioned and DP is orthogonal to the privacy protection mechanism of HE used in our paper. Also, DP inherently introduces noise or randomness to the model and may degrade the model accuracy, while Homomorphic Encryption (HE) generally does not hurt accuracy~\cite{boulemtafes2020review}. Finally, the idea of P3EFT (but not the name) is also briefly discussed in~\cite{sai2024parameter} as a potential research direction. 

The related studies of PEFT and PPFT have been presented in Section~\ref{secIntroduction}, and are not repeated here.

 
\subsection{Discussing Related Studies based on Figure \ref{fig:relationships}(b)} 

To present related research according to Fig.~\ref{fig:relationships}(b), where RA is short for ``Resource Allocation'', we will consider the following combinations: 1) RA + LoRA + FHE (i.e., RA of PrivTuner), 2) LoRA + FHE (i.e., PrivTuner), 3) RA + LoRA, and 4) RA + FHE.

As noted in the subsection above, only a few studies~\cite{zmushko2023privacy,bu2022differentially,sai2024parameter} touch upon P3EFT, and none of them is about RA, so there is no very 
related work on RA of PrivTuner or other P3EFT instances. 
Hence, except~\cite{zmushko2023privacy,bu2022differentially,sai2024parameter}, there is no other work closely related to PrivTuner. Yet, since PrivTuner is an important contribution of our paper, we present some work on ``FHE + Deep Learning (DL)'' from the perspective of different kinds of DL models. 
The DL models to be discussed include CNNs, RNNs,
and transformers.
\begin{itemize}[leftmargin=5pt] 
\item  
\textbf{Using FHE in CNNs.} CNNs are crypto-friendly due to linear structures and easily linearized activation layers. 
CryptoNets~\cite{gilad2016cryptonets} utilizes YASHE' for CNNs and achieves high accuracy on MNIST. CareNets~\cite{jin2019carenets} implements a compact CNN for inference on encrypted high-resolution images. 
SHE~\cite{lou2019she} uses TFHE to achieve a Shift-accumulation-based model, adapting ResNet-18 for the ImageNet dataset.
Lee~\textit{et~al}.~\cite{lee2022privacy} employ the RNS-CKKS~\cite{cheon2019full} to adapt ResNet-20 on the CIFAR-10 dataset.
\item  
\textbf{Using FHE in RNNs.}
RNNs often have higher multiplicative depth than CNNs, requiring larger crypto parameters~\cite{podschwadt2022survey}. Podschwadt~\textit{et~al}.~\cite{podschwadt2021non} propose parallel RNN blocks to reduce the multiplicative depth, which can be supported by CKKS.
Jang~\textit{et~al}.~\cite{jang2022privacy} implement encryption on GRUs using proposed encryption scheme MatHEAAN extended from CKKS.
\item 
\textbf{Using FHE in Transformers.} Transformers involve complex and nonlinear operations. Iron~\cite{hao2022iron} and BOLT~\cite{pang2023bolt} present privacy-preserving
inference solutions for transformers that support matrix multiplications and nonlinear computations. \mbox{BlindTune}~\cite{panzade2024can} further proposes to secure tuning transformers using CKKS and is tailored for image classification tasks. 
\end{itemize}



Resource allocation has been widely studied in the context of wireless communications, where it is often employed to optimize the use of network resources under different constraints. 
Previous works have applied RA to federated learning~\cite{zhou2022joint,zafar2024federated}, semantic communication~\cite{wang2024adaptive,li2023resource} and unmanned aerial vehicle (UAV) frameworks~\cite{yan2023joint,mao2023joint}, showing its effectiveness in improving overall network performance under different application scenarios.

Here we discuss when RA meets LoRA or PEFT in general. 
FedPipe~\cite{fang2024automated} proposes an automated federated pipeline that fine-tunes large language models using a novel MILP optimization approach, incorporating efficient local training and partial weight aggregation to enhance performance without increasing communication overhead.
Similarly, FedPEAT \cite{chua2023fedpeat} proposes a federated PEFT method while incorporating an adaptive control mechanism to optimize the latency, storage and performance. Yu~\textit{et~al}.\cite{yu2023orchestration}
also integrate the Emulator-Adapter architecture into Mobile Edge Computing (MEC) and employ a hybrid multi-agent Deep Reinforcement Learning (DRL) to ensure adaptability and fine-tuning efficiency.

Here we examine when RA meets FHE or PPTech in general. 
PFMLP \cite{fang2021privacy} proposes a multi-party privacy-preserving machine learning solution supported by the Paillier algorithm \cite{paillier1999public} and discusses the influence of encryption
key length, network structure and number of clients to optimize the computational overhead and training. DAMCREM \cite{suzuki2020damcrem} divides a task into several macro-tasks supported by FHE and proposes a dynamic allocation method of computation resources to macro-tasks to shorten the FHE execution time.






\begin{figure*}[t] 
\centering
\includegraphics[width=.85\linewidth]{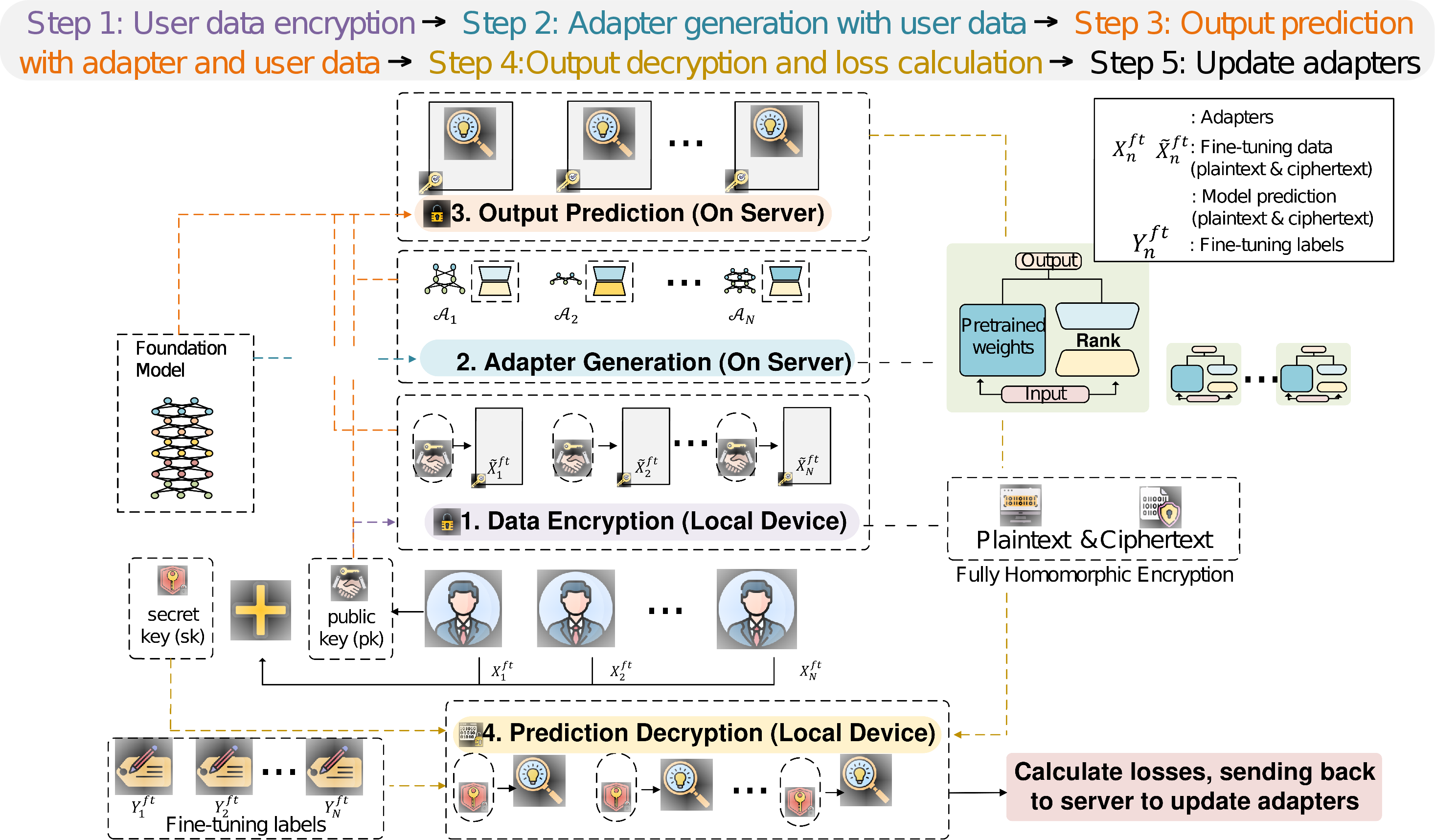}
\caption{An overview of the system and flows of our proposed PrivTuner.}
\vspace{-0.3cm}
\label{fig:system}
\end{figure*}

\section{System Model}\label{sec:system}
This section presents the system model, beginning with the threat model and an overview in Section~\ref{subsec:threat_model_overview}. Then, we elaborate on the proposed scheme of PrivTuner in Section~\ref{subsec:privtuner}. Fig. \ref{fig:system} presents the overall flows of PrivTuner.

\subsection{Threat Model and Overview}\label{subsec:threat_model_overview}

Consider a system comprising a model owner's server and a set of external mobile devices, denoted as $\mathcal{N}:=\{1,\ldots,N\}$.
The server hosts a pre-trained foundation model serving as a \textit{Machine Learning as a Service} (MLaaS) platform (e.g., ChatGPT operated by OpenAI).
Each mobile device $n\in \mathcal{N}$ aims to fine-tune the foundation model using its own private data to enhance personalized performance.
Operating under an honest but curious security model~\cite{goldreich2004foundations}, our PrivTuner allows devices to transmit FHE-encrypted data instead of the raw one to the server. 
Then, the server creates an adapter for each device via LoRA, and the inference results (also encrypted) are returned to the devices. 
Upon reception, devices decrypt the inference results and calculate the loss with the corresponding labels. The loss information is then sent back to the server for gradient calculation and adapter updating.



\subsection{HE-friednly BERT-Tiny Model}\label{subsec:FHEBERT}

In our work, we adopt the BERT-Tiny model, a small and efficient variant of the original BERT~\cite{turc2019well}. Following the approach proposed by Rovida \textit{et al.}~\cite{rovida2024transformer}, we use the CKKS scheme to perform encrypted computations while maintaining an acceptable trade-off between computational overhead and model accuracy.

The employed BERT-Tiny model involves large matrix multiplications and several non-linear layers (e.g., Softmax, GeLU, and LayerNorm). In particular, those non-linear functions can not be implemented directly, as the current CKKS only support addition and multiplication operations. 
To efficiently handle matrix multiplications in ciphertext, the SIMD technique supported by CKKS could be leveraged, which allows parallel computations on encrypted vectors. For different non-linear functions involved, we follow Rovida \textit{et al.}~\cite{rovida2024transformer} and employ different methods (e.g., Maclaurin series and Chebyshev polynomial) to approximate them efficiently. To save space, more detailed information and the approximation performance for non-linear layers are reported in Appendix.~\ref{sec:Approximations}. 

\subsection{PrivTuner} \label{subsec:privtuner}
Algorithm~\ref{Algorithm:PrivTuner} presents the sequence of steps executed by both mobile devices and the server within the PrivTuner scheme. We will detail each step in the following paragraphs.

\begin{algorithm}
\footnotesize
\caption{PrivTuner.} 
\label{Algorithm:PrivTuner}
\KwIn{
Pre-trained model parameterized by ${W}_0$; 
\newline
Sets of FHE parameters $(\bm{\lambda},\bm{q})$;
\newline
Sets of fine-tuning data and labels $(\bm{X}^{\textnormal{ft}},\bm{Y}^{\textnormal{ft}})$.
}
\KwOut{Set of fine-tuned adapter $\bm{A}^{\textnormal{ft}}$.}

\textbf{Step 1: Data Encryption.}\\
\For{\underline{each mobile device $n \in \mathcal{N}$ in parallel}}{

$\mathrm{KeyGen}(\lambda_n,q_n)$ 
$\rightarrow$  $\mathrm{pk}_n,\mathrm{sk}_n$; 

$\mathrm{Enc}(\mathrm{pk}_n,X_n^{\textnormal{ft}})$ $\rightarrow$ 
$\widetilde{X}_n^{\textnormal{ft}}$;

Transmit $\widetilde{X}_n^{\textnormal{ft}}$ and $\mathrm{pk}_n$ to the server;
}

\textbf{Step 2: Adapter Generation.}\\


Generate adapters: $\mathrm{LoRA}_n({W}_0)\rightarrow \mathcal{A}_n:=\{A_n^1,A_n^2\},~\forall n \in \mathcal{N}$; 







\textbf{Step 3: Prediction.}



Output predictions: 
$\widetilde{Y}_n^{\textnormal{p}}=\mathrm{Eval}\big({\mathrm{pk}}_n,({W}_0,\mathcal{A}_n),\widetilde{X}_n^{\textnormal{ft}},f^{p}\big),~\forall n \in \mathcal{N}$, where $f^{p}$ is the prediction function defined in (\ref{equa:predition});

Send result $\widetilde{Y}_n^{\textnormal{p}}$ back to each device $n$;

\textbf{Step 4: Decryption and Loss Computation.}

\For{\underline{each mobile device $n \in \mathcal{N}$ in parallel}}{

$\mathrm{Dec}(\mathrm{sk}_n,\widetilde{Y}_n^{\textnormal{p}})$ 
$\rightarrow$ $Y_n^{\textnormal{p}}$; 

Compute the loss: $\mathcal{L}_n(Y_n^{\textnormal{p}},Y_n^{\textnormal{ft}})$ $\rightarrow$ $L_n$;

Transmit $L_n$ to the server;
}

\textbf{Step 5: Update of Adapters.}

Update each adapter $\mathcal{A}_n$ with $L_n$ and obtain $\mathcal{A}^{\textnormal{ft}}_n:=\{A_n^{1,\textnormal{ft}},A_n^{2,\textnormal{ft}}\}$.
\end{algorithm}
\textbf{Data Encryption.} Each mobile device $n$ initially generates a unique set of keys of FHE for encryption:
\begin{align}
    {\mathrm{pk}}_n, {\mathrm{sk}}_n = \mathrm{KeyGen}(\lambda_n,q_n),~\forall n \in \mathcal{N},\label{equa:key_generation1}
\end{align}
where $\lambda_n$ is the polynomial degree and $q_n$ is the coefficient modulus. 
With the generated public key $\mathrm{pk}_n$ from (\ref{equa:key_generation1}), now device $n$ could encrypt its fine-tuning data $X_n^{\textnormal{ft}}$ as follows:
\begin{align}
    {\widetilde{X}_n^{\textnormal{ft}}}= \mathrm{Enc}({\mathrm{pk}}_n,X_n^{\textnormal{ft}}),~\forall n \in \mathcal{N},\label{equa:data}
\end{align}
where $\widetilde{X}_n^{\textnormal{ft}}$ is the encrypted data. Both $\widetilde{X}_n^{\textnormal{ft}}$ and $\mathrm{pk}_n$ are then transmitted to the server.

\textbf{Adapter Generation.}
Upon receiving $\widetilde{X}_n^{\textnormal{ft}}$ and $\mathrm{pk}_n$, the server could apply various PEFT methods to the pre-trained foundation model $W_0$. Here we employ LoRA~\cite{hu2021lora} (explained in Section~\ref{subsec:lora}) to generate an adapter $\mathcal{A}_n=\{A_n^1,A_n^2\}$ for each device $n$ as described below:
\begin{align}
   \mathcal{A}_n = \mathrm{LoRA}_n(W_0),~\forall n \in \mathcal{N},\label{equa:lora}
\end{align}
where 
$\mathrm{LoRA}_n(\cdot)$ denotes the low-rank adaption operation (we have the subscript ``$n$'' just for generality, and the experiments use the same adapter for all $n$).
Since the size of $\mathcal{A}_n$ is much smaller than that of $W_0$, 
this approach dramatically reduces the number of trainable parameters.

\textbf{Prediction.} Then, the server could perform predictions on the encrypted data $\widetilde{X}_n^{\textnormal{ft}}$ with $(W_0,\mathcal{A}_n)$ and $\mathrm{pk}_n$, generating the corresponding predictions $\widetilde{Y}_n^{\textnormal{p}}$ (also encrypted). As discussed in Section~\ref{subsec:FHEBERT}, we could follow Rovida \textit{et al.}~\cite{rovida2024transformer} to implement the prediction, and the process can be generalized as follows:
\begin{align}
\widetilde{Y}_n^{\textnormal{p}}=\mathrm{Eval}\big({\mathrm{pk}}_n,({W}_0,\mathcal{A}_n),\widetilde{X}_n^{\textnormal{ft}},f^{p}\big),~\forall n \in \mathcal{N},
\end{align}
where $\mathrm{Eval}(\cdots)$ has been discussed in Section~\ref{secfhe} and $f^{p}$ denotes the prediction function working as below:
\begin{align}
    f^{p}({W}_0,\mathcal{A}_n,\widetilde{X}_n^{\textnormal{ft}})=W_0\widetilde{X}_n^{\textnormal{ft}}+A_n^1A_n^2\widetilde{X}_n^{\textnormal{ft}},~\forall n \in \mathcal{N}.\label{equa:predition}
\end{align}
where $W_0$ is the foundation model, $\mathcal{A}_n$ is the generated adapter, and $\widetilde{X}_n^{\textnormal{ft}}$ is the encrypted dataset for fine-tuning.
Since the fine-tuning labels $\bm{Y}^{\textnormal{ft}}$  are kept by mobile devices, the server sends $\widetilde{Y}_n^{\textnormal{p}}$ back to each device $n$ for loss computation.
It is worth noting that using prediction $\widetilde{Y}_n^{\textnormal{p}}$ to reason about the parameter matrix $W_0$ is computationally infeasible, especially for devices with limited computing resources. Thus, this will not compromise security.

\textbf{Decryption and Loss Computation.} Upon receiving $\widetilde{Y}_n^{\textnormal{p}}$, each device $n$ first use the secret key $\mathrm{sk}_n$ to perform decryption and obtain the prediction $Y_n^{\textnormal{p}}$ in plaintext:
\begin{align}
    Y_n^{\textnormal{p}} = \mathrm{Dec}(\mathrm{sk}_n,\widetilde{Y}_n^{\textnormal{p}}),~\forall n \in \mathcal{N}.
\end{align}
Then, the corresponding loss $L_n$ is computed as follows:
\begin{align}
    L_n = \mathcal{L}_n(Y_n^{\textnormal{p}},Y_n^{\textnormal{ft}}),~\forall n \in \mathcal{N},
\end{align}
where $\mathcal{L}_n(\cdot)$ represents the employed loss function, and we do not specify it here for generality. The loss $L_n$ is transmitted to the server for gradient calculation and adapter updating.

\textbf{Updates.} Given the loss $L_n$ and weight matrices $W_0,\mathcal{A}_n=\{A_n^1,A_n^2\}$, the server can compute the gradient and update $\mathcal{A}_n$ into fine-tuned $\mathcal{A}_n^{\textnormal{ft}}$ following a standard fine-tuning process~\cite{panzade2024can}. 
For an arbitrary input $X_n^{\textnormal{input}}$, now the server could provide the prediction service as:
\begin{align}
    Y_n^{\textnormal{input}} = W_0X_n^{\textnormal{input}} + A_n^{1,\textnormal{ft}}A_n^{2,\textnormal{ft}}X_n^{\textnormal{input}},~\forall n \in \mathcal{N},
\end{align}
where $Y_n^{\textnormal{input}}$ is the prediction result after fine-tuning.

\subsection{Security Analysis}
Here, we preliminarily evaluate the security of PrivTuner. On the server side, the foundational model is retained by itself, with only prediction results shared with external devices, thus safeguarding model privacy.  On the mobile device side, private data is encrypted prior to transmission, ensuring there is no data exposure during fine-tuning, even in adversarial contexts. 
Losses are in plaintext but do not compromise privacy, as the server only handles ciphertext results.
We will delve deeper into potential attacks on FHE cryptosystems in Section~\ref{subsec:privacy_anlaysis}.

\section{Consumption and Privacy Protection Modelling} \label{secConsumption}
This section models the consumption and privacy protection within PrivTuner. Sec.~\ref{subsec:consumption_device} and Sec.~\ref{subsec:consumption_server} model the energy consumption on mobile devices and the server. Sec.~\ref{subsec:privacy_anlaysis} analyzes privacy protection against cryptographic attacks.

\subsection{Consumption on Mobile Devices}\label{subsec:consumption_device}
The consumption on mobile devices is mainly caused by the encryption operation in (\ref{equa:data}) and the transmission of encrypted data. 
For ease of analysis, we set coefficient moduli $\bm{q}$ in FHE to fixed large values to ensure sufficient arithmetic depth, while polynomial degrees $\bm{\lambda}$ are allowed to be adjusted as needed.

\textbf{Encryption Time.}
Let $D_n$ be the total number of samples in data $X_n^{\textnormal{ft}}$, and $s_n$ be the number of tokens per sample that need to be encrypted. Then the encryption time $T_{n}^{\textnormal{en}}$ is:
\begin{align}
    T_{n}^{\textnormal{en}} = \frac{y_1(\lambda_n)D_ns_n}{g_n},~\forall n \in \mathcal{N},
\end{align}
where $y_1(\lambda_n)$ denotes the required CPU cycles for encrypting per token under the polynomial modulus $\lambda_n$, and $g_n$ is the allocated computing capacity on device $n$.

\textbf{Encryption Energy.} Following~\cite{dinh2020federated},  we further define the energy consumption for encryption as follows:
\begin{align}
    E_{n}^{\textnormal{en}}=\kappa y_1(\lambda_n)s_nD_ng_n^2,~\forall n \in \mathcal{N},
\end{align}
where $\kappa$ is the effective switched capacitance that indicates the energy efficiency of the computation.

For the transmission periods, since the ciphertext data $\widetilde{X}_n^{\textnormal{ft}}$ is usually much larger than the plaintext loss $L_n$, we focus on the transmission of the previous one. Following~\cite{zhou2023resource}, we adopt Frequency Division Multiple Access (FDMA) for wireless transmission. In this context, the total available bandwidth is represented as $B^{\textnormal{total}}$. The bandwidth allocated to user $n$ is $B_n$ and we have $\sum_{n=1}^N B_n \leq B^{\textnormal{total}}$.
According to the Shannon formula, we could obtain the uplink transmission rate $r_n$ from user $n$ to the server as follows:
\begin{align}
r_n=B_n\log_2(1+\frac{p_nh_n}{N_0B_n}),~\forall n \in \mathcal{N},
\end{align}
where $B_n$ is the bandwidth as noted, $p_n$ is the transmission power of user $n$, $h_n$ is the channel attenuation from  user $n$ to the server, and $N_0$ is the noise spectral density. 

\textbf{Transmission Time.} The transmission time is given by:
\begin{align}
    T_{n}^{\textnormal{tr}} = \frac{d_{n}}{r_n},~\forall n \in \mathcal{N},
\end{align}
where $d_{n}$ (bits) is the size of the transmitted data. 

\textbf{Transmission Energy.} The transmission energy is:
\begin{align}
    E_{n}^{\textnormal{tr}} = p_nT_{n}^{\textnormal{tr}},~\forall n \in \mathcal{N}.
\end{align}

\subsection{Consumption on Server}\label{subsec:consumption_server}
The server in a network often houses abundant communication resources. Hence, our focus is primarily on the computational consumption of the server, i.e., the encrypted prediction process in Step 3 of Algorithm \ref{Algorithm:PrivTuner}.

Given that the prediction is supported by FHE, the computational demands on ciphertexts are inherently higher than those on plaintexts. This motivates the need to model and analyze the consumption metrics for further optimization. Unfortunately, a critical challenge we have to face is:
\begin{challenge}
    How to quantitatively measure the time and energy consumption for predictions when different computations are performed on ciphertexts?
\end{challenge}

Inspired by~\cite{marcano2019fully},
we estimate the total time and energy required for predictions by calculating the number of CPU cycles needed for various operations on ciphertexts. Specifically, we give the following definition:

\begin{definition}\label{def:1}
    The major operations involved in model prediction are addition, multiplication and rotation, accompanied by other customized operations, e.g., bootstrapping.
    Hence, the total number of CPU cycles $C^{\textnormal{total}}$ could be measured by:
    \begin{align}
        C^{\textnormal{total}} = C^{\textnormal{add}} + C^{\textnormal{mul}} +
        C^{\textnormal{rot}} +
        C^{\textnormal{other}},
    \end{align}
    where $C^{\textnormal{add}}$ accounts for the CPU cycles of all additions, $C^{\textnormal{mul}}$ for all multiplications, $C^{\textnormal{rot}}$ for all rotations, and $C^{\textnormal{other}}$ for cycles required by other operations.
\end{definition}

Since the server knows the model structure and the utilized FHE techniques, it can roughly count the number of additions, multiplications and rotations required for prediction. Hence, with $a_{n}$, $m_{n}$ and $o_n$ denoting the count of additions, multiplications, and rotations needed per sample, we estimate CPU cycles as follows:
\begin{align}
    y_2(\lambda_n) \!= \!y_3(\lambda_n)a_{n} \!+ \!y_4(\lambda_n)m_{n} \!+ \!y_5(\lambda_n)o_n \!+ \!C_n^{\textnormal{other}},~\forall n \in \mathcal{N}, \label{eqy2lambda}
\end{align}
where $y_2(\lambda_n)$ denotes the total CPU cycles under $\lambda_n$. Besides, $y_3(\lambda_n)$, $y_4(\lambda_n)$ and $y_5(\lambda_n)$ correspond to cycles for each addition, multiplication and rotation, respectively. $C_n^{\textnormal{other}}$ represents CPU cycles for other operations.
The expressions of $y_3(\lambda_n)$, $y_4(\lambda_n)$ and $y_5(\lambda_n)$ are discussed in Sec.~\ref{subsec:expression}.

\textbf{Computation Time.} Given $y_2(\lambda_n)$ as the required CPU cycles for the prediction process in step 2 of Algorithm \ref{Algorithm:PrivTuner}, and let $c_n$ be the CPU cycles for adapter updating in step 4 of Algorithm \ref{Algorithm:PrivTuner}, the total computation time on the server for device $n$ is defined as:
\begin{align}
    T_n^{\textnormal{cmp}} = \frac{y_{2}(\lambda_n)D_n}{f_n},~\forall n \in \mathcal{N},\label{equa:T_cmp}
\end{align}
where $D_n$ is the number of samples and $f_n$ denotes the allocated frequency to device $n$ on the server. 

\textbf{Computation Energy.} The computation energy is :
\begin{align}
E_n^{\textnormal{cmp}} =  \kappa y_2(\lambda_n) D_n f_n^2,~\forall n \in \mathcal{N}.\label{equa:E_cmp}
\end{align}

\subsection{Privacy Protection Level}\label{subsec:privacy_anlaysis}


To assess defenses against adversaries in FHE, we employ a metric of \textbf{security level}~\cite{albrecht2021homomorphic}, defined as the estimated number of operations (in bits) required to breach the cryptographic security. We consider three attack types: uSVP \cite{schnorr2003lattice}, BDD \cite{liu2013solving}, and hybrid dual \cite{albrecht2017dual}. The privacy protection of an FHE setting $(\lambda,q)$ is determined by the minimum security level across these attacks, each measured by the LWE-estimator \cite{albrecht2015concrete}.

However, owing to the intricacy of the computational model used by the LWE-estimator, directly quantifying the impact of an FHE setting on privacy protection is challenging. Recall that $q_n$ for each device $n$ has been fixed for ease of analysis, we define a function $y_5(\cdot)$ to describe the relationship between $\lambda_n$ and
privacy protection $\mathcal{S}_n$:
\begin{align}
    \mathcal{S}_n = y_6(\lambda_n),~\forall n \in \mathcal{N},
\end{align}
where $y_6(\lambda_n)$ is discussed in Sec.~\ref{subsec:expression}. 
We also consider that different devices could have different privacy concerns.
 With $\sigma_n$ serving as a weight parameter denoting the concern for the privacy protection of device $n$,
the overall privacy level is quantified as the sum of the individual ones of each device:
\begin{align}
\mathcal{S}^{\textnormal{total}}=\sum_{n=1}^N\sigma_n\mathcal{S}_n.
\end{align}

\section{Problem Formulation}\label{sec:problem}
In this section on problem formulation, Sec.~\ref{secOptimization}  presents a joint optimization problem to balance energy and privacy. Then, in Sec.~\ref{subsec:expression}, the expressions of functions used in the optimization problem are specified. Finally, in Sec.~\ref{secChallenges}, we discuss the challenges of solving the optimization problem.

\subsection{The Studied Optimization Problem} \label{secOptimization}
The proposed PrivTuner enhances security but also results in considerable energy consumption, especially for mobile devices with limited resources. 
Therefore, one of our goals is to minimize the overall energy consumption:
\begin{align}\label{obj:1}
\sum_{n=1}^N (E_{n}^{\textnormal{en}} +E_n^{\textnormal{tr}}+E_n^{\textnormal{cmp}}).
\end{align}\\
Concurrently, we also aim to maximize the privacy protection of the overall system, which can be expressed as:
\begin{align}\label{obj:2}
    \sum_{n=1}^N\sigma_n\mathcal{S}_n.
\end{align}

Synthesizing the above two objectives (\ref{obj:1}) and (\ref{obj:2}), a joint optimization problem can be formulated as follows:
\begin{subequations}\label{Original_Problem0}
\begin{align}
\hspace{75pt}&\hspace{-75pt}\mathbb{P}_0:  \min_{\bm{f},\bm{g},\bm{p},\bm{B},\bm{\lambda}}
\left\{\sum_{n=1}^N(E_{n}^{\textnormal{en}} \!+\!E_n^{\textnormal{tr}}\!+\!E_n^{\textnormal{cmp}})\!-\!\omega \sum_{n=1}^N\sigma_n\mathcal{S}_n\right\},
\tag{\ref{Original_Problem0}} \\
    \text{s.t.}~
    &  \sum_{n=1}^N f_{n} \leq f^{\textnormal{total}},\label{Constra:f} \\
    &  \sum_{n=1}^N B_{n} \leq B^{\textnormal{total}},\label{Constra:B} \\
    &  g_n \leq g^{\textnormal{max}}_n,\forall n \in \mathcal{N}, \label{Constra:g}\\
    &  p_n \leq p^{\textnormal{max}}_n,\forall n \in \mathcal{N}, \label{Constra:p}\\
    &  \lambda_n \in \{\lambda_{\textnormal{o}1},\ldots,\lambda_{\textnormal{o}m}\},\forall n \in \mathcal{N},\label{Constra:lambda}\\
    &  T_{n}^{\textnormal{en}}+T_{n}^{\textnormal{tr}} \leq T^{\textnormal{max}}_{D},\forall n \in \mathcal{N}, \label{Constra:T_s_max_0}\\
    &  T_n^{\textnormal{cmp}}\leq T^{\textnormal{max}}_{S},\forall n \in \mathcal{N},\label{Constra:T_u_max_0}
\end{align}
\end{subequations}
where $\boldsymbol{f},\boldsymbol{g},\boldsymbol{p},\boldsymbol{B}$ and $\boldsymbol{\lambda}$ are optimization variables. $\boldsymbol{f}=[f_1,\ldots,f_N]$ is a vector containing the allocated frequency to each device at the server, $\boldsymbol{g}=[g_1,\ldots,g_N]$ contains the frequency on each device, $\boldsymbol{p}=[p_1,\ldots,p_N]$ denotes the transmission power of devices, $\boldsymbol{B}=[B_1,\ldots,B_N]$ denotes the allocated bandwidth, and $\boldsymbol{\lambda}=[\lambda_1,\ldots,\lambda_n]$ represents the employed polynomial degrees. Besides, $\omega$ in (\ref{Original_Problem0}) is the weight parameter to balance the energy consumption and privacy protection. Constraints (\ref{Constra:f}) and (\ref{Constra:g}) limit the available computation frequencies allocated to device $n$ at the server and device $n$, respectively. 
Constraints (\ref{Constra:B}) and (\ref{Constra:p}) limit the available communication resources, i.e., transmission power and bandwidth.
Constraint (\ref{Constra:lambda}) provides $m$ choices for the polynomial modulus degrees: $\{\lambda_{\textnormal{o}1},\ldots,\lambda_{\textnormal{o}m}\}$, and ranks them in ascending order. Constraints (\ref{Constra:T_s_max_0}) and  (\ref{Constra:T_u_max_0}) set the maximum time budgets for both devices and the server. 

\subsection{Expressions of Functions used in the Optimization Problem}\label{subsec:expression}

To conduct a comprehensive analysis of the solution to $\mathbb{P}_0$, it is essential to specify the expressions of functions $y_1(\lambda_n)$, $y_3(\lambda_n)$, $y_4(\lambda_n)$, $y_5(\lambda_n)$, and $y_6(\lambda_n)$ (note that $y_2(\lambda_n)$ uses $y_3(\lambda_n)$, $y_4(\lambda_n)$ and $y_5(\lambda_n)$, as given in~(\ref{eqy2lambda})). We first establish the following Condition~\ref{condition:fitting} for those functions:
\begin{condition}\label{condition:fitting}
    The functions $y_1(\lambda_n)$, $y_3(\lambda_n)$, $y_4(\lambda_n)$, and $y_5(\lambda_n)$ are all convex with respect to the variable $\lambda_n$.
\end{condition}
\begin{remark}
    This assertion of convexity is based on empirical observations in \cite{murugesan2021analysis} and our experiments in Section~\ref{sec:experimental results}. 
    Convex forms simplify analysis and show satisfactory performance.
\end{remark}

We explored various convex forms for fitting.
For $y_1(\lambda_n)$, our experiments in Section~\ref{sec:experimental results} suggest a quadratic expression:
\begin{align}
    y_1(\lambda_n) = C_0(\lambda_n+C_1)^2,~\forall n \in \mathcal{N},
\end{align}
where $C_0,C_1>0$ are constants obtained from the curve fitting experiment, to be elaborated in Section~\ref{sec:experimental results}.

Regarding the ciphertext addition function $y_3(\lambda_n)$ and multiplication function $y_4(\lambda_n)$, experimental findings indicate a linear growth of computation cycles with respect to $\lambda_n$: 
\begin{align}
    y_3(\lambda_n) &= C_2\lambda_n+C_3,~\forall n \in \mathcal{N},\\
    y_4(\lambda_n) &= C_4\lambda_n+C_5,~\forall n \in \mathcal{N},\\
    y_5(\lambda_n) &= C_6\lambda_n+C_7,~\forall n \in \mathcal{N},
\end{align} 
where $C_2,C_4,C_5 > 0$, and $C_3,C_5,C_7$ are all constants derived from curve fitting. Their specific values are also presented in Section~\ref{sec:experimental results}. 

For function $y_6(\lambda_n)$, we also give the following condition:

\begin{condition}\label{condition:security}
    The function $y_6(\lambda_n)$ is concave regarding $\lambda_n$.
\end{condition}
\begin{remark}
    The concavity not only facilitates the solution but also means diminishing marginal gains, a characteristic commonly seen in diverse practical applications \cite{kumari2019fair}.
\end{remark}

Via experiments in Section~\ref{sec:experimental results}, we find constants $C_8>0$ and $C_9$ to curve-fit the privacy protection function $y_6(\lambda_n)$:
\begin{align}
    y_6(\lambda_n) = C_8\lambda_n+C_9,~\forall n \in \mathcal{N}.
\end{align} 


\subsection{Challenges of Solving the Optimization Problem $\mathbb{P}_0$ of (\ref{Original_Problem0})} \label{secChallenges}
$\mathbb{P}_0$ is difficult to solve due to the following two challenges:
\subsubsection{NP-Hardness}
The optimization variables in problem $\mathbb{P}_0$ form a highly coupled, inseparable mixed-integer non-linear programming (MINLP) problem, with discrete variable $\bm{\lambda}$ and continuous variables $(\bm{f}, \bm{g}, \bm{p}, \bm{B})$, making it NP-hard.

\subsubsection{Non-Convexity}
The terms $E_n^{\textnormal{en}}$ and $E_n^{\textnormal{ft}}$ in the objective function (\ref{Original_Problem0}) include non-convex products, while the term $E_n^{\textnormal{tr}}$ incorporates non-convex divisions, which can be verified by their Hessian matrices.  In addition, $T_n^{\textnormal{en}}$ in constraint (\ref{Constra:T_s_max_0}) and $T_n^{\textnormal{ft}}$ in (\ref{Constra:T_u_max_0}) are also non-convex ratios. The above non-convex terms result in the overall non-convexity of $\mathbb{P}_0$.


\section{Algorithm Design}\label{sec:algorithm}


This section introduces an iterative algorithm devised to solve $\mathbb{P}_0$ by alternately optimizing two different sets of variables. Each iteration comprises two stages below:
\begin{itemize}
    \item Stage 1: Given fixed $(\bm{p},\bm{B})$, optimize $(\bm{f},\bm{g},\bm{\lambda})$.
    \item Stage 2: Given obtained  $(\bm{f},\bm{g},\bm{\lambda})$, optimize $(\bm{p},\bm{B})$.
\end{itemize}
Stage 1 in Section~\ref{subsec:optimization_cmp} optimizes computation energy and privacy protection while transmission-related variables are fixed. Stage 2 in Section~\ref{subsec:optimization_tr} further optimizes the remaining variables based on the outcomes from Stage 1.
The overall RA algorithm and its analysis are detailed in Section~\ref{subsec:analysis}.

\subsection{Stage 1: Optimization in Computation Energy and Privacy Protection Level}\label{subsec:optimization_cmp}
In Stage 1, we first focus on the optimization of $(\bm{f},\bm{g},\bm{\lambda})$ to optimize computation energy consumption and privacy protection level. 
Given fixed $(\bm{p},\bm{B})$, the original problem $\mathbb{P}_0$ could be simplified into $\mathbb{P}_1$:
\begin{subequations}\label{Subproblem1}
\begin{align}
\mathbb{P}_1:
\min_{\bm{f},\bm{g},\bm{\lambda}}~&
\sum_{n=1}^N \Big(\kappa y_1(\lambda_n)D_ns_ng_n^2
+ \kappa y_2(\lambda_n)D_nf_n^2\Big)\nonumber\\
&- \omega\sum_{n=1}^N \sigma_n y_6(\lambda_n),
\tag{\ref{Subproblem1}} \\
\text{s.t.}~ & 
\text{(\ref{Constra:f})},~\text{(\ref{Constra:g})},~\text{(\ref{Constra:lambda})}\nonumber\\
& \frac{y_1(\lambda_n)D_ns_n}{g_n}+T_n^{\textnormal{tr}}\leq T^{\textnormal{max}}_D 
, ~\forall n \in \mathcal{N},\label{Constra:T_s_max}\\
& \frac{y_2(\lambda_n)D_n}{f_n}\leq T^{\textnormal{max}}_{S} 
, ~\forall n \in \mathcal{N}.\label{Constra:T_u_max}
\end{align}
\end{subequations}
Since we aim to minimize the computation energy consumption in (\ref{Subproblem1}), computing capacities  $\bm{f},\bm{g}$ need to be minimized as much as possible, i.e., the time budgets $T^{\textnormal{max}}_{D}$ and $T^{\textnormal{max}}_{S}$ in (\ref{Constra:T_s_max}) and (\ref{Constra:T_u_max}) will be exhausted. Thus, we can actually make constraints (\ref{Constra:T_s_max}) and (\ref{Constra:T_u_max}) into equations and obtain the solution of $\bm{f},\bm{g}$ as functions with regard to $\bm{\lambda}$:
\begin{subequations}
\begin{align}
&\overline{g}_n(\lambda_n) = \frac{y_1(\lambda_n)D_ns_n}{T^{\textnormal{max}}_D-T_n^{\textnormal{tr}}},~\forall n \in \mathcal{N}, \label{rela:f,lambda} \\
&\overline{f}_n(\lambda_n) =  \frac{y_2(\lambda_n)D_n}{T^{\textnormal{max}}_S},~\forall n \in \mathcal{N}. \label{rela:g,lambda}
\end{align}
\end{subequations}
It is evident that $\overline{f}_n(\lambda_n)$ in (\ref{rela:f,lambda}) and $\overline{g}_n(\lambda_n)$ in (\ref{rela:g,lambda}) are both monotonically increasing of $\lambda_n$.
Furthermore, their convexity is also established by verifying their Hessian matrices. Thus, constraints (\ref{Constra:f}) and (\ref{Constra:g}) are equivalent to:
\begin{subequations}
\begin{align}
& \sum_{n=1}^N \overline{f}_n(\lambda_n) \leq f^{\textnormal{total}},\label{Constra:lambda_1}\\
&\lambda_n \leq \lambda_{n}^{\textnormal{max}} = \sqrt{\frac{({T^{\textnormal{max}}_D-T_n^{\textnormal{tr}}})g_n^{\textnormal{max}}}{C_1D_ns_n}}-C_2. \label{Constra:lambda_2}
\end{align}
\end{subequations}
Substituting (\ref{rela:f,lambda}), (\ref{rela:g,lambda}), (\ref{Constra:lambda_1}) and (\ref{Constra:lambda_2}) into problem $\mathbb{P}_1$, we can obtain an equivalent problem $\mathbb{P}_2$ with only one variable $\bm{\lambda}$: 
\begin{align}
\mathbb{P}_2:
\min_{\bm{\lambda}}&
\sum_{n=1}^{N}\kappa y_1(\lambda_n)D_ns_n \overline{g}_n^2(\lambda_n)
-\omega \sum_{n=1}^{N}\sigma_ny_5(\lambda_n)
\nonumber\\
&+\sum_{n=1}^{N}\kappa y_2(\lambda_n)D_n\overline{f}_n^2(\lambda_n),
\label{Subproblem1:v1} \\
\text{s.t.}~ & 
(\ref{Constra:lambda}),
(\ref{Constra:lambda_1}),(\ref{Constra:lambda_2}).\nonumber
\end{align} 

Until now, $\mathbb{P}_2$ becomes an Integer Programming (IP) problem. Hence, we employ a \textbf{Branch and Bound (B\&B) algorithm}~\cite{zoppei2022branch} to solve it effectively. 
The whole progress to utilize the B\&B algorithm to solve $\mathbb{P}_2$ is generalized in \textbf{Algorithm~\ref{Algorithm:BB}}. 

\begin{algorithm}
\caption{B\&B Algorithm for problem $\mathbb{P}_2$} 
\label{Algorithm:BB}
\footnotesize
\KwIn{
INLP Problem $\mathbb{P}_2$ with objective function $f_{\mathbb{P}}()$ 
}
\KwOut{Optimal solution $\bm{\lambda}^*$.}

Relax the integer constraint (\ref{Constra:lambda}) of $\mathbb{P}_2$ to create $\mathbb{P}_3$;

Set $\mathcal{L} = \{\mathbb{P}_3\}$ and  $f_{\mathbb{P}}(\lambda^*)=\infty$; 

\While{$\mathcal{L} \neq \emptyset$}{
Select a problem $P(0)$ from $\mathcal{L}$ to explore;

Call (\ref{suboptimal:lambda}) to solve $P(0)\rightarrow \bm{\lambda}'$;

\uIf{$\bm{\lambda}'$ satisfies (\ref{Constra:lambda}) \textbf{and} $f_{\mathbb{P}}(\bm{\lambda}')<f_{\mathbb{P}}(\bm{\lambda}^*)$}{
Update $\bm{\lambda}^* \leftarrow \bm{\lambda}'$;
 }
\uElseIf{$\bm{\lambda}'$ is infeasible \textbf{or} $f_{\mathbb{P}}(\bm{\lambda}')\geq f_{\mathbb{P}}(\bm{\lambda}^*)$}{
    Fathom $P(0)$ from $\mathcal{L}$;
}
\Else{
\textit{Branch} $P(0)$ into $P(1)$, $P(2)$;

$\mathcal{L}=\mathcal{L} \cup \{P(1),P(2)\}$;

Fathom $P(0)$ from $\mathcal{L}$;
}

\For{$P$ in $\{P(1), P(2)\}$}{
Call (\ref{suboptimal:lambda}) to solve $P\rightarrow \bm{\lambda}'$;

\If{$f_{\mathbb{P}}(\bm{\lambda}')<f_{\mathbb{P}}(\bm{\lambda}^*)$}
{
Fathom $P$ from $\mathcal{L}$; 
}
}
}
\end{algorithm}

Since the B\&B algorithm is already well-researched, we do not give details here, but the key operations in Algorithm \ref{Algorithm:BB}: \textit{Problem Relaxation} (Line 2), \textit{Solution of the Relaxed Problem} (Line 6) and \textit{Branch} (Line 12), are detailed in \textbf{Appendix A}.

With obtained optimal $\bm{\lambda}^*$ from Algorithm~\ref{Algorithm:BB}, 
the optimal $\bm{f}^*$ and $\bm{g}^*$ can also be derived from (\ref{rela:f,lambda}) and (\ref{rela:g,lambda}):
\begin{align}
    &f^*_n = \frac{y_1(\lambda_n^*)D_ns_n}{T^{\textnormal{max}}_D-T_n^{\textnormal{tr}}},~\forall n \in \mathcal{N}, \label{optimal:f}\\
    &g^*_n =  \frac{y_2(\lambda_n^*)D_n}{T^{\textnormal{max}}_S},~\forall n \in \mathcal{N}.\label{optimal:g}
\end{align}

\subsection{Stage 2: Optimization in Transmission Energy}\label{subsec:optimization_tr}
In Stage 2, we focus on optimizing transmission-related variables $(\bm{p},\bm{B})$ to minimize the transmission energy consumption. 
With given $(\bm{f},\bm{g},\bm{\lambda})$, problem $\mathbb{P}_0$ could be simplified into $\mathbb{P}_4$:
\begin{subequations}\label{Subproblem2}
\begin{align}
\mathbb{P}_4: 
\min_{\bm{p},\bm{B}}~&\sum_{n=1}^N \frac{p_nd_{n}}{r_n},
\tag{\ref{Subproblem2}} \\
\text{s.t.}~
&(\ref{Constra:p}), (\ref{Constra:B}),\nonumber\\
& r_n^{\textnormal{min}} \leq r_n, \forall n \in \mathcal{N},\label{Constra:r}
\end{align}
\end{subequations}
where $r_n^{\textnormal{min}}=\frac{d_{n}}{T_S^{\textnormal{max}}-T_n^{\textnormal{en}}}$.
However, problem $\mathbb{P}_4$ is a difficult sum-of-ratio problem as validated by Lemma \ref{lemma1}:
\begin{lemma}\label{lemma1}
    For problem $\mathbb{P}_4$, we have:
    \begin{itemize}
        \item Rate $r_n$ is jointly concave with respect to $(p_n,B_n)$ ;
        \item $\frac{p_nd_{n}}{r_n}$ in (\ref{Subproblem2}) is a jointly pseudoconvex ratio.
    \end{itemize}
\end{lemma}

\begin{proof}
    We calculate the Hessian matrix of $r_n$:
    \begin{align}
    H = 
    \begin{bmatrix}
    -\frac{h_n^2}{B_nN_0^2{(1+\frac{p_nh_n}{N_0B_n})}^2\ln{2}} & 
    \frac{p_n^D h_n^2}{B_n^2N_0^2{(1+\frac{p_ng_n}{N_0B_n})}^2\ln{2}} \\
    \frac{p_n h_n^2}{B_n^2N_0^2{(1+\frac{p_n^{2}g_n}{N_0B_n})}^2\ln{2}} & 
    -\frac{p_n^2h_n^2}{B_n^3N_0^2{(1+\frac{p_nh_n}{N_0B_n})}^2\ln{2}} 
    \end{bmatrix}
    \notag
\end{align}
$H$ is a negative semidefinite matrix, proving $r_n$ is jointly concave with $(p_n^D,B_n)$. The numerator $p_nd_{n}$ is an affine function. Thus, $\frac{p_nd_{n}}{r_n}$ is deemed as a pseudoconvex ratio according to Page~245 of the book~\cite{cambini2008generalized}. Lemma \ref{lemma1} is proved. 
\end{proof}

\begin{remark}
Based on Lemma \ref{lemma1}, term $\frac{p_nd_{n}}{r_n}$ is identified as jointly pseudoconvex. However, the sum of pseudoconvex functions may not retain pseudoconvexity, complicating the resolution of sum-of-ratio problems~\cite{jong2012efficient}.
\end{remark}

To efficiently tackle problem $\mathbb{P}_4$, we utilize an advanced fractional programming method \cite{zhao2023human} to transform it into an equivalent problem $\mathbb{P}_5$:
\begin{subequations}\label{Subproblem2_v1}
\begin{align}
\mathbb{P}_5:
\min_{ \bm{p}\bm{B}}~&\sum_{n=1}^N\big((p_nd_n)^2z_n+\frac{1}{4(r_n)^2z_{n}}\big),
\tag{\ref{Subproblem2_v1}} \\
\text{s.t.}~
& (\ref{Constra:p}),(\ref{Constra:B}),(\ref{Constra:r}). \nonumber
\end{align}
\end{subequations}
where $\bm{z}:=[z_1,\ldots,z_{N}]$ is the introduced auxiliary variable. It is worth noting that problem $\mathbb{P}_5$ is a convex problem now and can be solved by analyzing its KKT Conditions. 
The process of utilizing this method to update $\bm{z}$ and solve $\mathbb{P}_5$ is listed in \textbf{Algorithm \ref{Algorithm:Fractional_Programming}}. It starts by calculating the initial values of auxiliary variables $\bm{z}^{(0)}$, then alternatively refines $(\bm{p},\bm{B})$ and $\bm{z}$ until convergence or reaching maximum iteration number $I$.

\begin{algorithm}
\footnotesize
\caption{Fractional Programming Method} 
\label{Algorithm:Fractional_Programming}
\KwIn{
Initial feasible solution $(\bm{p}^{(0)},\bm{B}^{(0)})$;
\newline
Maximum iteration number $I$;
\newline
Error parameter $\epsilon \in (0,1)$.
}
\KwOut{Optimal solution 
$(\bm{p}^{*},\bm{B}^{*})$.}

\SetKwProg{MyFunction}{Function}{:}{end}
\MyFunction{\underline{Fractional\_Programming$(\bm{p}^{(0)},\bm{B}^{(0)},I,\epsilon)$}}{
Compute initial values of  $\bm{z}^{(0)}$, where:
\begin{align}
    z_n^{(0)} &= \frac{1}{2p_n^{(0)}d_{n}r_n\big(p_n^{(0)},B_n^{(0)}\big)}, \nonumber
\end{align}

Set iteration counter $i=1$.

\While{\underline{$i \leq I$ and not convergence}}{
Solve $\mathbb{P}_5$ to obtain $(\bm{p}^{(i)},\bm{B}^{(i)})$ according to \textbf{Theorem \ref{theorem:optimal_p_B}} on page \pageref{pagetheorem:optimal_p_B} with given $\bm{z}^{(i-1)}$.

Update auxiliary variables $\bm{z}^{(i)}$, where:
\begin{align}
    z_n^{(i)} = \frac{1}{2p_n^{(i)}d_{n}r_n\big(p_n^{(i)},B_n^{(i)}\big)},~\forall n \in \mathcal{N}.\nonumber
\end{align}

If $|\bm{z}^{(i)}-\bm{z}^{(i-1)}| \leq \epsilon $, the algorithm terminates.

Let $i \leftarrow i+1$.
}
}
\end{algorithm}
Until now, what we need to focus on is how to solve problem $\mathbb{P}_4$ when $\bm{z}$ is already given based on Algorithm \ref{Algorithm:Fractional_Programming}. We provide the following Theorem \ref{theorem:optimal_p_B} to obtain the optimal $(\bm{p}^*,\bm{B}^*)$:

\begin{theorem}\label{theorem:optimal_p_B}
    The optimal solution $(p_n^*,B_n^*)$ of problem $\mathbb{P}_4$ can be computed by suitably defined functions $\overline{B}_n(\cdot,\cdot)$ and $\overline{p}_n(\cdot,\cdot)$ such that:
    \begin{align}
        B_n^* &= \overline{B}_n(\beta^*,\gamma_n^*) ~(\text{i.e., (\ref{equa:optimal_B})}), \\
        p_n^{*} &=  \overline{p}_n(B_n^*,\gamma_n^*) ~(\text{i.e., (\ref{equa:optimal_p})}), 
    \end{align}
    where $\gamma_n^*$ is obtained from (\ref{optimal:gamma}) and $\beta^*$ is the solution to (\ref{temp:2}).
\end{theorem}\label{pagetheorem:optimal_p_B}
\begin{proof}
    The proof of Theorem \ref{theorem:optimal_p_B} is in Appendix B.
\end{proof}

\subsection{A Review of Proposed Resource Allocation Algorithm}\label{subsec:analysis}

The overall joint resource allocation algorithm is presented in Algorithm \ref{algo:resourceallocation}. It starts by initializing a feasible solution, and then alternatively optimizing $(\bm{f},\bm{g},\bm{\lambda})$ and $(\bm{p},\bm{B})$ until convergence or reaching maximum iteration number.

\begin{algorithm}
\footnotesize
\caption{Joint Resource Allocation Algorithm} 
\label{algo:resourceallocation}
\KwIn{Initial $sol^{(0)}\!=\!(\bm{p}^{(0)},\bm{B}^{(0)},\bm{f}^{(0)},\bm{g}^{(0)},\bm{\lambda}^{(0)})$;
\newline
Maximum iteration number $J$; 
\newline
Error parameter $\epsilon \in (0,1)$.}

\KwOut{Optimal solution $sol^{*}\!=\!(\bm{p}^{*},\bm{B}^{*},\bm{f}^{*},\bm{g}^{*},\bm{\lambda}^{*})$.}

\SetKwProg{MyFunction}{Function}{:}{end}

\MyFunction{\underline{Joint\_Optimization $(sol^{(0)},J,\epsilon)$}}{

Set iteration counter $j=1$.

\While{\underline{$j \leq J$ and not convergence}}{

Given $(\bm{p}^{(j-1)},\bm{B}^{(j-1)})$, solve problem $\mathbb{P}_5$ based on Algorithm \ref{Algorithm:BB}, (\ref{optimal:f}) and (\ref{optimal:g}) to obtain $(\bm{f}^{(j)},\bm{g}^{(j)},\bm{\lambda}^{(j)})$.

Given $(\bm{f}^{(j)},\bm{g}^{(j)},\bm{\lambda}^{(j)})$, solve $\mathbb{P}_5$ based on Algorithm \ref{Algorithm:Fractional_Programming} and Theorem \ref{theorem:optimal_p_B} 
to obtain  $(\bm{p}^{(j)},\bm{B}^{(j)})$.

Update the solution:
\begin{align}
sol^{(j)} \leftarrow (\bm{p}^{(j)},\bm{B}^{(j)},\bm{f}^{(j)},\bm{g}^{(j)},\bm{\lambda}^{(j)}). \nonumber
\end{align}

If $|sol^{(j)}-sol^{(j-1)}| \leq \epsilon $, the algorithm terminates.

Set $j \leftarrow j+1$.
}
}
\end{algorithm}

\textbf{Time Complexity.}
We provide a time complexity analysis of Algorithm \ref{algo:resourceallocation} focusing on the number of devices $N$. The complexity primarily stems from two stages in each iteration: Lines 4 and 5.
Line 4 is to solve for  $(\bm{f},\bm{g},\bm{\lambda})$ by calling the B\&B algorithm, (\ref{optimal:f}) and (\ref{optimal:g}).
In B\&B algorithms, with $m$ branches from each node and $N$ as the solution depth, the worst-case computational complexity is $\mathcal{O}(m^N)$, as same as the brute force. However, the pruning operation effectively reduces the actual computation time. If we utilize $K$ as the estimated running times, and each solution actually takes $\mathcal{O}(N)$, the total time complexity could be regarded as $\mathcal{O}(3KN)$. The solution of $(\bm{p},\bm{B})$ in (\ref{optimal:f}) and (\ref{optimal:g}) takes $\mathcal{O}(2N)$.
Line $5$ is to alternatively update $(\bm{z})$ and $(\bm{p},\bm{B})$ by calling Algorithm \ref{Algorithm:Fractional_Programming} and Theorem \ref{theorem:optimal_p_B}.
Theorem \ref{theorem:optimal_p_B} incurs a complexity of 
$\mathcal{O}(3N)$ from the calculation of multipliers and solutions.
For Algorithm \ref{Algorithm:Fractional_Programming}, given maximum iteration number $I$, Line 2 and Lines 4-9 contribute $\mathcal{O}(N)$ and $\mathcal{O}(I(4N))$, respectively. Hence, the total complexity of Line 5 is $\mathcal{O}(4(I+1)N)$. 
Finally, with the maximum iteration number $J$, we can conclude the overall time complexity of Algorithm \ref{algo:resourceallocation} is $\mathcal{O}(J(3K+4I+6)N)$. 

\textbf{Solution Quality and Convergence.}
Algorithm \ref{algo:resourceallocation} iteratively optimizes two sets of variables, $(\bm{\lambda},\bm{f},\bm{g})$ and $(\bm{p},\bm{B})$, by alternately solving problems $\mathbb{P}_2$ and $\mathbb{P}_5$. Algorithm \ref{Algorithm:BB} ensures the global optimum for $\mathbb{P}_2$. Algorithm \ref{Algorithm:Fractional_Programming} and Theorem \ref{theorem:optimal_p_B} guarantee global optima to $\mathbb{P}_5$. Consequently, while Algorithm \ref{algo:resourceallocation} does not guarantee a globally optimal solution to the original problem $\mathbb{P}_0$, it does secure the global optimality of each subproblem's solution. The precision parameter 
$\epsilon$ in Algorithm \ref{algo:resourceallocation} ensures that the solution converges to the desired level of accuracy. The convergence is also clear from the above analysis.

\section{Experiments}\label{sec:experimental results}

\subsection{Macrobenchmark}
We implement the CKKS-supported BERT-Tiny in C++ and Python, utilizing the OpenFHE~\cite{OpenFHE} library. Specifically, we evaluate the experiments on a Linux x86\_64 machine equipped with Intel(R) Xeon(R) Gold 5218R CPU@2.10GHz. In addition, we also use HEXL~\cite{boemer2021intel} to accelerate OpenFHE. Following Rovida \textit{et al.}~\cite{rovida2024transformer}, we configure the polynomial degree $\lambda=2^{15}$, coefficient modulus $q=1767$ bits, with a precision factor of $55$ bits. 
We evaluated the model accuracy on three GLUE datasets: SST-2, MRPC and RTE. The batch size is set as $32$; the learning rate is $1\times10^{-5}$, and
the training epoch is $10$. The accuracy performance is reported in Table~\ref{Tab:accuracy_fhe}. 
Besides, we also provide a computation runtime breakdown when doing inference on one sample of $8$ tokens from the SST-2 dataset, which is as shown in Table~\ref{Tab:overheads_fhe}. 

\begin{table}[h]
  \caption{Test performance on different GLUE datasets.}
  \label{Tab:accuracy_fhe}
  \centering
  \resizebox{\linewidth}{!}{ 
  \begin{tabular}{l|c|c|c}
    \hline
    Dataset& BERT-Tiny &FHE-BERT-Tiny &Loss~$\downarrow$ \\\hline
    SST-2 (Acc.) &0.823 &0.790 &0.033 \\
    MRPC (Acc.) &0.703 & 0.675 &0.028 \\
    WNLI (Acc.) &0.601 &0.564 &0.037\\
    CoLA (Acc., M Corr.) &0.691, 0 &0.691, 0 &0.000, 0\\
    \hline
  \end{tabular}
}
\end{table}


\begin{table}[h]
  \caption{A computation runtime breakdown for inference.}
  \label{Tab:overheads_fhe}
  \centering
  \resizebox{.75\linewidth}{!}{ 
  \begin{tabular}{l|cc|c}
    \hline
      Operation& Time Consumption (s) \\\hline
    Encryption (client)  &0.7106     \\
    Prediction (server)  &163.3211   \\
    Decryption (client)  &0.0119    \\\hline
    Total  &164.0436   \\\hline
  \end{tabular}
  } 
\end{table}




Besides, since we aim to optimize the polynomial degree $\lambda$ in CKKS to balance the overhead and security, we also report the runtime and security level under different $\lambda$.
From Table~\ref{Tab:overheads_comparison}, we can see that a higher $\lambda$ increases the runtime obviously, but also benefits the security level. We also want to emphasize that even if the overall runtime is often hundreds of seconds, the primary bottleneck lies in the server-side encrypted predictions. In contrast, client-side tasks are limited to encryption and decryption, keeping the computational demands relatively manageable for mobile devices (Table~\ref{Tab:overheads_fhe}).

\begin{table}[h]
  \caption{Runtime and security level under different $\lambda$.}
  \label{Tab:overheads_comparison}
  \centering
  \resizebox{.9\linewidth}{!}{
  \begin{tabular}{l|c|cc}
    \hline
      Method & $\lambda$ &Runtime & Security Level\\\hline
    \multirow{3}{*}{FHE-BERT-Tiny}  &$2^{15}$ &164.04 s & 66.1\\
      &$2^{16}$ &330.13 s & 128.4 \\
     &$2^{17}$ &719.64 s & 277.0 \\\hline
    PEFT (LoRA) &-  &0.067 s &0 (unprotected) \\\hline
\end{tabular}
}
\end{table}

\subsection{Curve Fitting}\label{subsec:curve}
Here, we test the number of CPU cycles required for different operations: encryption, addition,  multiplication and rotation. It is also worth noting that CKKS is not the only option, and our analysis can be adapted to other FHE mechanisms. Specifically, we measure the number of required CPU cycles by obtaining the average running time and computation frequency. The curve fitting results are shown in Fig. \ref{fig:FHE}(a).
To fit the privacy protection level, we run the LWE-estimator under three attacks and give the results in Table~\ref{table:security}. The minimum one of three attacks is adopted as the corresponding privacy protection. Fig.~\ref{fig:FHE}(b) presents the curve fitting performance.

\begin{figure}[!t]
    \centering
    \begin{subfigure}[b]{0.48\linewidth}
        \centering
\includegraphics[width=\linewidth]{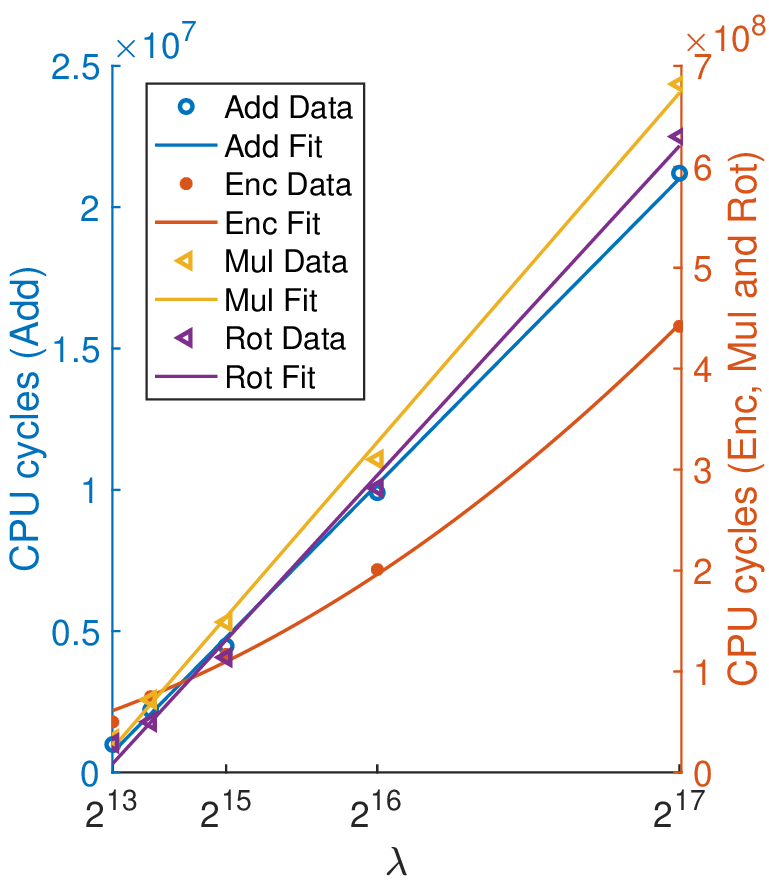}
    \caption{The CPU cycles for each Addition (Add), Encryption (Enc), Multiplication (Mul) and Rotation (Rot) under different $\lambda$.}
    \end{subfigure}
    \hspace{3pt}
    \begin{subfigure}[b]{0.48\linewidth}
        \centering
\includegraphics[width=\linewidth]{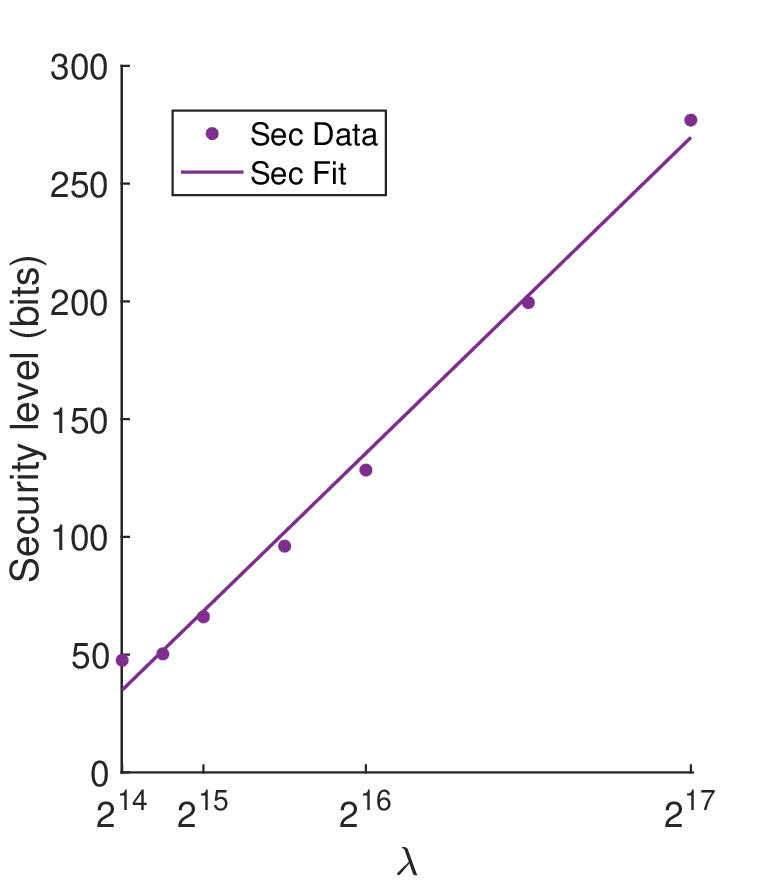}
    \caption{Privacy protection level measured by the LWE-estimator under different $\lambda$ from Table \ref{table:security} and the fitting performance.}
    \end{subfigure}
    \caption{Fitting results under different polynomial degrees $\lambda$.}
    \label{fig:FHE}
\end{figure}

\begin{table}[ht]
\centering
\caption{
Estimated running times
(in bits) for three attacks under different $\lambda$ with $q=1767$ bits in CKKS.}
\label{table:security}
\begin{tabular}{@{}c|ccc@{}}
\toprule
{Polynomial modulus $\lambda$}& {uSVP}~\cite{schnorr2003lattice} & {BDD}~\cite{liu2013solving} & {hybrid dual}~\cite{albrecht2017dual}\\
\midrule
16384    & \textbf{47.6} &  \textbf{47.6} & 48.6  \\
24576    & 50.5 &  \textbf{50.3} & 50.9 \\
32768    & 66.3 &  \textbf{66.1} & 66.6  \\
49152    & \textbf{96.1} &  \textbf{96.1} & 96.7  \\
65536    & 128.5 &  \textbf{128.4} & 129.4  \\
98304    & 199.7 &  \textbf{199.5} & 200.7  \\
131072    & \textbf{277.0} &  296.6 & 278.1  \\
\bottomrule
\end{tabular}
\end{table}
The fitted functions and corresponding constants obtained from the curve fitting are listed in Table \ref{tab:function_fitting}.
\begin{table}[h]
    \centering
    \caption{Result of fitted functions.}
    \label{tab:function_fitting}
    \begin{tabular}{l|l}
    \hline
        Functions & 
        Expressions of fitted functions
        \\ \hline
        Encryption & 
        $C_0(\lambda+C_1)^2$, where  $C_0=0.012$, $C_1=6.45\times {10}^4$
          \\ 
        Addition&  
        $C_2\lambda+C_3$, where  $C_2=165.15$, $C_3=-6.34\times {10}^5$
          \\ 
        Multiplication &
        $C_4\lambda+C_5$, where  $C_4=5282.55$, $C_5=-1.91\times {10}^7$
          \\ 
        Rotation &
        $C_6\lambda+C_7$, where  $C_6=4979.69$, $C_7=-3.19\times {10}^7$
          \\ \hline
        Security level &
        $C_8\lambda+C_9$, where  $C_8=0.0020$, $C_9=1.4789$
          \\\hline
    \end{tabular}
\end{table}

\subsection{Parameter Settings}\label{subsec:setting}
We provide a detailed description of the experiments' default parameter settings. The total number of mobile devices $N$ is $10$. The number of polynomial degree options $M$ is $3$ and we have $\{\lambda_{\textnormal{o}1},\lambda_{\textnormal{o}2},\lambda_{\textnormal{o}3}\}=\{2^{15},2^{16},2^{17}\}$. 
The number of samples $D_n$ is 16, and the number of tokens $s_n$ for encryption is 10.
The estimated number of additions $a_n$, multiplications $m_n$, and rotations $o_n$ required for prediction are set to $1660$, $260$ and $1460$ (Please see Appendix~\ref{Appendix:count} for details).

The number of CPU cycles $C_n^{\textnormal{other}}$ for other operations is set as $10^9$. The size of encrypted data for transmission $d_n$ is $3\times10^9$ bits.
The privacy concern $c_n$ of each device is randomly selected from $\{10,5,1\}$, and the weight parameter   $\omega$ for the overall privacy level is set to $1$ (unless configured otherwise). 

Then, we introduce settings regarding computation and communication resources. The total server computation capacity $f^{\textnormal{total}}$ is $20$ GHz. The maximum device frequency $g_n^{\textnormal{max}}$ is $3$ GHz.
The path loss model is $128.1+37.6\log(\text{distance})$, with an $8$ dB standard deviation for shadow fading, and the unit of distance is kilometer. The noise power spectral density $N_0$ is $-174$ dBm/Hz. The effective switched capacitance $\kappa$ is $10^{-28}$. 
The total available band width  $B^{\textnormal{total}}$ is $10$ MHz. The maximum device transmission power $p_n^{\textnormal{max}}$ is $20$ dBm.
The maximum time budget on devices $T^{\textnormal{max}}_D$ is $1000$ s, and the time budget on the server $T^{\textnormal{max}}_S$ is set as $3000$ s (unless configured otherwise). 

\subsection{Comparison with Benchmarks on Energy Optimization}\label{subsec:comparison}
In this subsection, we mainly explore the proposed algorithm's performance in terms of energy optimization. We first introduce three benchmarks as follows:
\begin{enumerate}
    \item \textbf{Average Allocation:} 
    In this case, we set up the resources of the server and mobile devices fairly. Specifically, we set each $f_n,g_n,p_n,B_n$ as $f^{\textnormal{total}}/N,g_n^{\textnormal{max}}/2,p_n^{\textnormal{max}}/2,B^{\textnormal{total}}/N$, respectively.
    Besides, the choices of $\bm{\lambda}$ are all set to $\lambda_{\textnormal{o}1}$.
    \item  \textbf{Optimize $\bm{f},\bm{g}$ only:} Here we fix each $p_n,B_n$ and $\lambda_n$ be $p_n^{\textnormal{max}}/2,B^{\textnormal{total}}/N$ and $\lambda_{\textnormal{o}1}$, respectively. Then, we only optimize variables $\bm{f},\bm{g}$ as discussed in Sec.~\ref{subsec:optimization_cmp}.
    \item  \textbf{Optimize $\bm{p},\bm{B}$ only:} Here we configure each  $f_n, g_n$ and $\lambda_n$ as $f^{\textnormal{total}}/N,g_n^{\textnormal{max}}/2$ and $\lambda_{\textnormal{o}1}$, respectively. Then, we only optimize variables $\bm{p},\bm{B}$ as discussed in Sec.~\ref{subsec:optimization_tr}.
\end{enumerate}
For a fair comparison, we temporarily set the weight parameter $\omega$ for privacy protection in the proposed algorithm to 0, i.e., all $\lambda_n$ are also $\lambda_{\textnormal{o}1}$. Then we conduct simulations under the following scenarios with different resources or budgets.

\begin{figure*}
    \centering
    \begin{subfigure}[b]{0.28\linewidth}
        \centering
    \includegraphics[width=\linewidth]
        {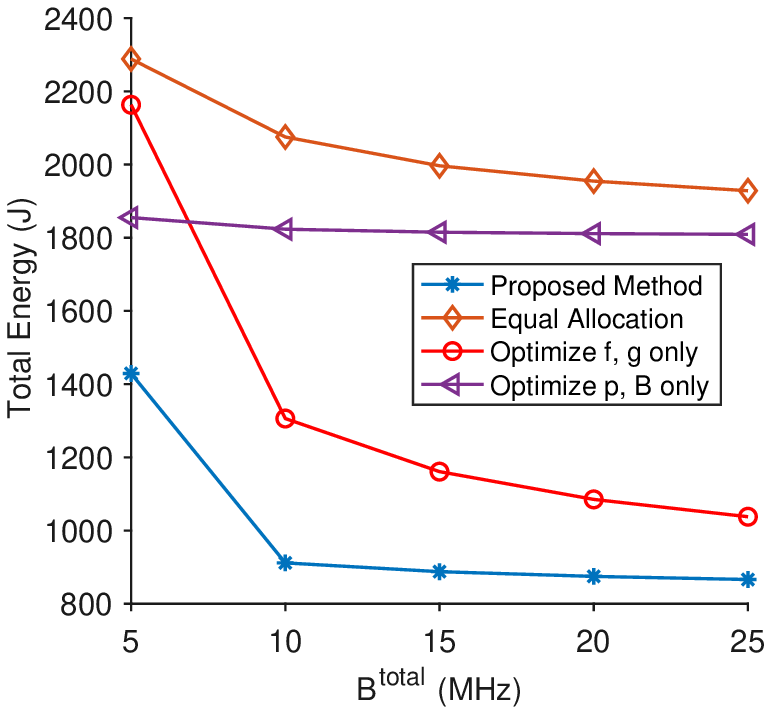}
        \caption{}
    \end{subfigure}
    \hspace{20pt}
    \begin{subfigure}[b]{0.28\linewidth}
        \centering
        \includegraphics[width=\linewidth]
        {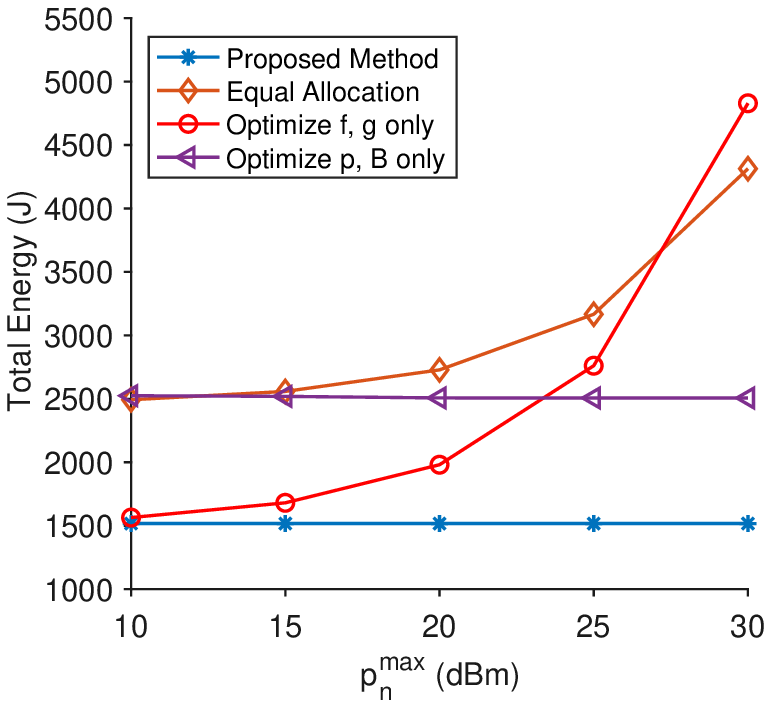}
    \caption{}
    \end{subfigure}
    \hspace{20pt}
    \begin{subfigure}[b]{0.28\linewidth}
        \centering
    \includegraphics[width=\linewidth]
        {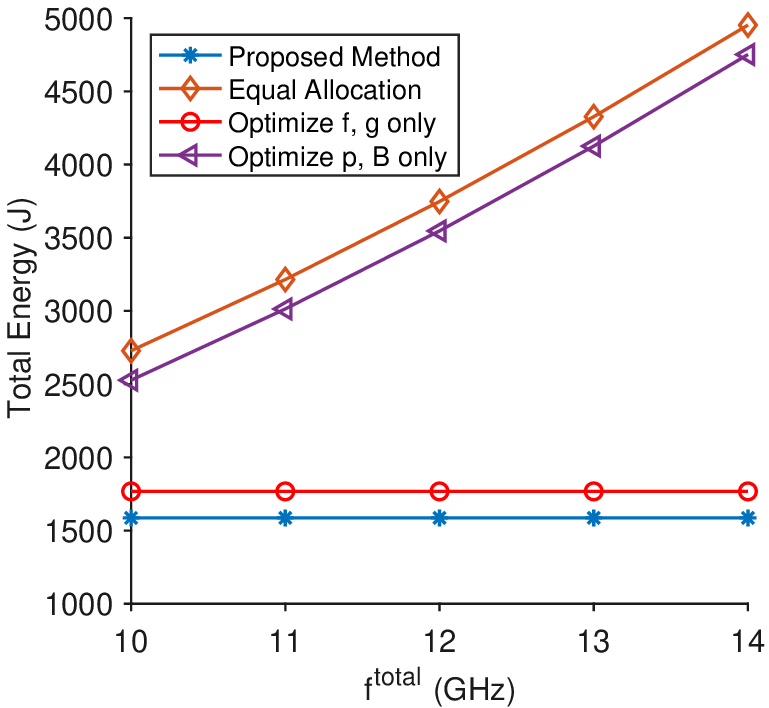}
        \caption{}
    \end{subfigure}\\
    \begin{subfigure}[b]{0.28\linewidth}
        \centering
        \includegraphics[width=\linewidth]
        {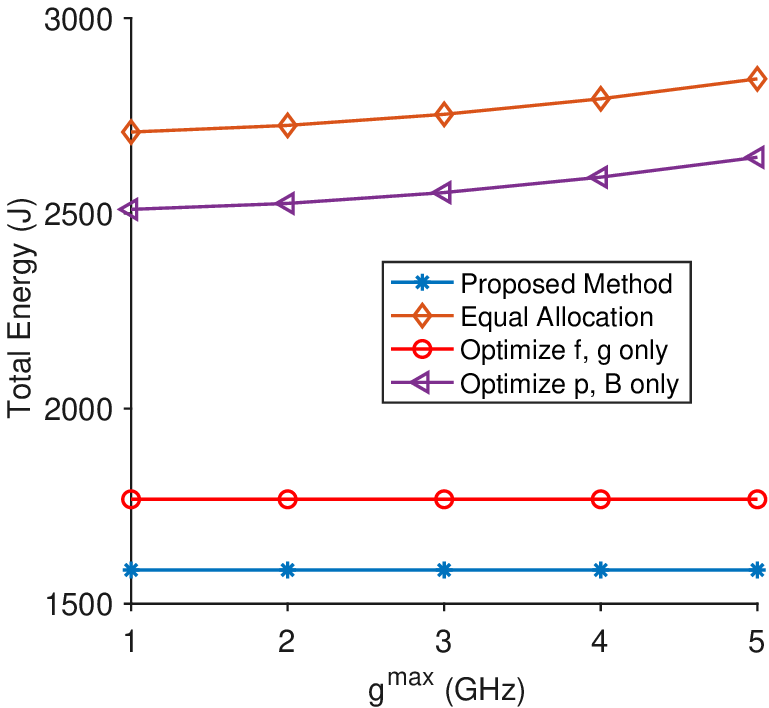}
    \caption{}
    \end{subfigure}
    \hspace{20pt}
    \begin{subfigure}[b]{0.28\linewidth}
        \centering
    \includegraphics[width=\linewidth]
        {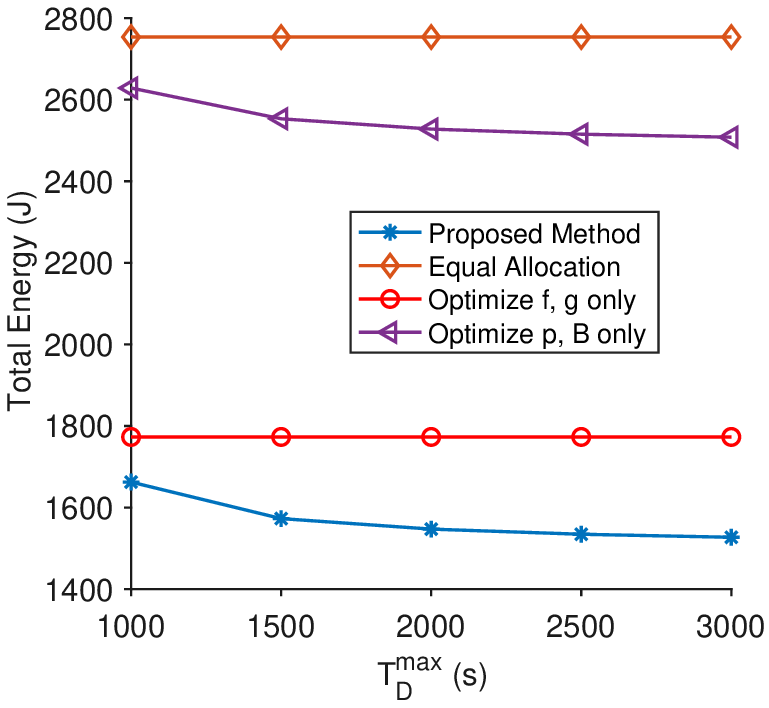}
        \caption{}
    \end{subfigure}
    \hspace{20pt}
    \begin{subfigure}[b]{0.28\linewidth}
        \centering
        \includegraphics[width=\linewidth]
        {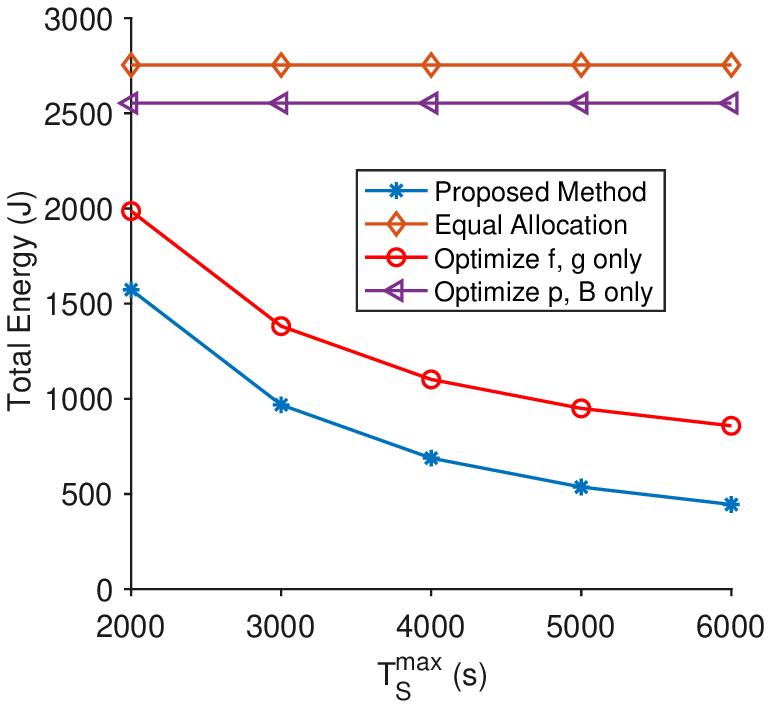}
    \caption{}
    \end{subfigure}
    \caption{Comparison with benchmarks on energy optimization under different resource budgets.}
    \label{fig:consumption}
    \vspace{-10pt}
\end{figure*}

\textbf{Total Bandwidth.} Here we vary the total available bandwidth $B^{\textnormal{total}}$ from 5 MHz to 25 MHz. In Fig.~\ref{fig:consumption}(a), larger available bandwidth means higher transmission rates, thus reducing uplink transmission latency and energy consumption. 
In addition, our proposed method always retains the best optimization performance compared with all benchmarks.

\textbf{Transmission Power.} Here we vary all the mobile devices' maximum transmission powers $p_n^{\textnormal{max}}$ from $10$ dBm to $30$ dBm. In Fig.~\ref{fig:consumption}(b), a higher $p_n^{\textnormal{max}}$ induces larger energy consumption if not optimized. It is worth noting that the optimization of ``Optimize $\bm{p}$, $\bm{B}$ only" and our proposed method remains stable since the optimal power solution has already been found and increasing the search range does not bring changes.

\textbf{Total Server Frequency.} Here we vary the total server CPU frequency $f^{\textnormal{total}}$ from $10$ GHz to $14$ GHz. Since the optimal frequencies need to be minimized to bring lower energy, we can see from Fig.~\ref{fig:consumption}(c) that curves of ``Optimize $\bm{f}$, $\bm{g}$ only" and our proposed method do not change as $f^{\textnormal{total}}$ increases.

\textbf{Mobile Device Frequency.}
Here we vary all mobile devices' maximum frequencies $g_n^{\textnormal{max}}$ from $1$ GHz to $5$ GHz. Intuitively, a trend similar to that in Fig.~\ref{fig:consumption}(c) can be observed in Fig.~\ref{fig:consumption}(d) since devices' frequencies also need to be minimized. Our proposed method still outperforms other benchmarks.

\textbf{Time Budget on Devices.} Here we configure the time budget on devices $T_D^{\textnormal{max}}$ from $1000$ s to $3000$ s. A higher $T_D^{\textnormal{max}}$ allows mobile devices to adopt lower computation frequencies and transmission powers to further decrease energy consumption. Hence, a slower downward trend can be observed in Fig.~\ref{fig:consumption}(e).

\textbf{Time Budget on Server.} Here we vary the time budget on server $T_S^{\textnormal{max}}$ from $2000$ s to $6000$ s. A higher $T_S^{\textnormal{max}}$ allows the server to allocate lower computation frequencies $\bm{f}$ to mobile devices, thus decreasing energy consumption. Thereafter, the energy optimization of our proposed method and ``Optimize $\bm{f}$, $\bm{g}$ only" decreases as $T_S^{\textnormal{max}}$ grows.

\subsection{Impact of Weight Parameter on Energy and Privacy}\label{subsec:weight}
Now we check how the weight parameter $\omega$ affects the tradeoff between energy and overall privacy level.
We also detail the choices of $\bm{\lambda}$ by devices with different levels of privacy protection $c_n$. 
Specifically, we raise time budgets $T^{\textnormal{max}}_D$ and $T^{\textnormal{max}}_S$ to $1500$ s and $5000$ s, and we adjust $\omega$ from $1$ to $10$. The experimental results are reported in Fig.~\ref{fig:omega}(a) and (b).

Fig.~\ref{fig:omega}(a) shows that as $\omega$ increases, both energy consumption and privacy level rise. This trend happens because a higher $\omega$ makes overall privacy level more significant in the joint optimization. As a result, our method tends to use a larger $\bm{\lambda}$ to improve privacy, sacrificing energy optimization performance.

Fig.~\ref{fig:omega}(b) shows how devices with different privacy concerns select $\bm{\lambda}$.
Recall that we categorize devices into three levels of privacy concern $c_n$ from $\{10, 5, 1\}$,
and we select one device from each level to analyze how they adjust $\lambda_n$ as $\omega$ varies.
They all demonstrate a step-wise increase in $\lambda_n$ as $\omega$ increases, and devices with higher $c_n$ values tend to choose larger $\lambda_n$, aligning with their greater contribution to overall privacy.

\begin{figure}[!t]
    \centering
    \begin{subfigure}[b]{0.48\linewidth}
        \centering
\includegraphics[width=\linewidth]{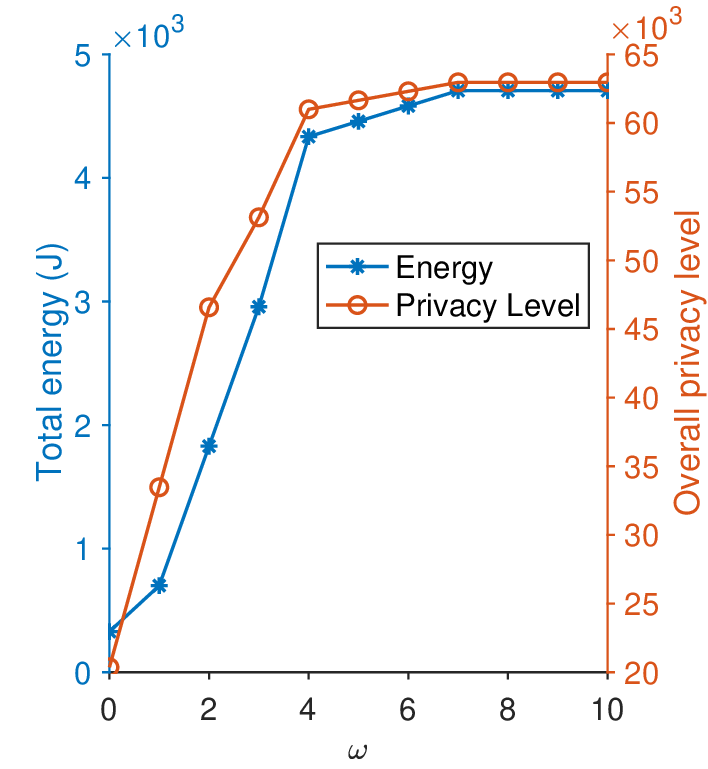}
    \caption{Total Energy $E^{\textnormal{total}}$ and overall privacy level $S^{\textnormal{total}}$ under different weight parameter $\omega$.}
    \end{subfigure}
    \hspace{3pt}
    \begin{subfigure}[b]{0.48\linewidth}
        \centering
\includegraphics[width=\linewidth]{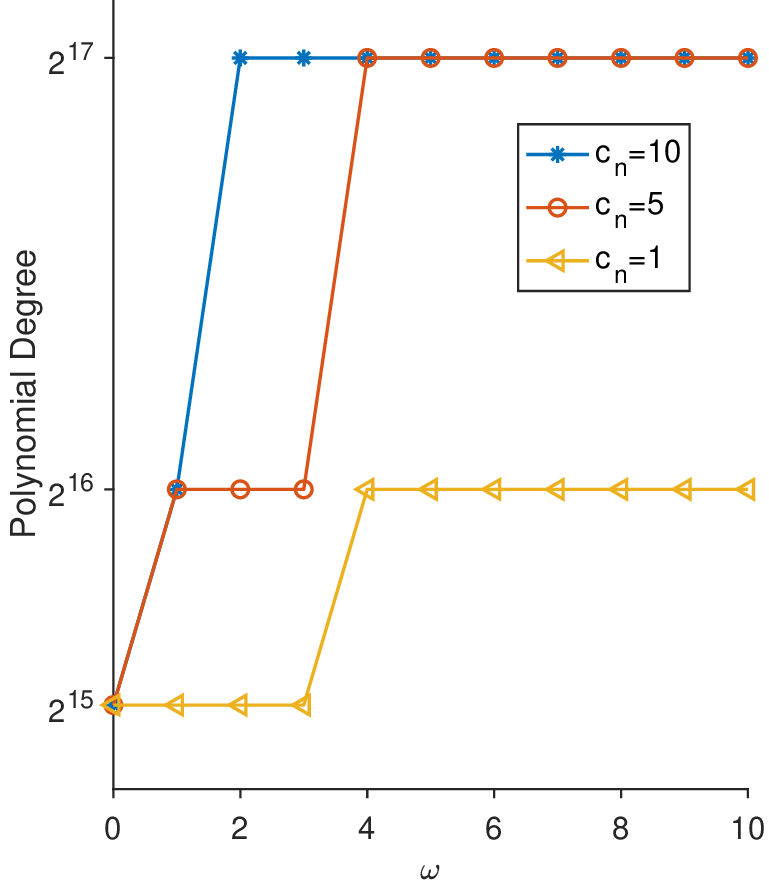}
    \caption{Choices of polynomial degree $\lambda_n$ by devices with different privacy concern $c_n=\{10,5,1\}$.}
    \end{subfigure}
    \caption{Optimization performance and choices of polynomial degrees $\lambda_n$ under different weight parameters $\omega$.} 
    \label{fig:omega}
    \vspace{-10pt}
\end{figure}

\section{Discussion \& Conclusion}\label{sec:7}

\textbf{Privacy of Adapters.} One potential concern for PrivTuner is
that the adapter may leak some information about data.
We use unencrypted adapters to facilitate the following inference services after fine-tuning. 
Given this limitation, we adopt an honest but curious security model, where the server is assumed to refrain from malicious actions (e.g., data reconstruction attacks). Besides, 
adapters are much smaller than the foundation model and contain far less information.
We also consider potential solutions in future work, including fine-tuning in a trustworthy environment~\cite{turtiainen2023trusted,shepherd2024trusted}, or using MPC techniques~\cite{zhao2019secure,dong2023puma} to distribute updates across multiple servers without any single server having full adapter information.

\textbf{Malicious Security.}
Supporting malicious security poses challenges for our proposed PrivTuner. If either the server or the client maliciously deviates from the protocol (e.g., by intentionally sending incorrect data), then the fine-tuning process may be compromised. 
However, even under a malicious adversary model, PrivTuner maintains security by isolating client and server roles. A malicious client gains no extra model information and only risks its own process if deviating, while a malicious server cannot deduce true labels from plaintext loss due to encrypted predictions. Additionally, client-specific FHE secret keys prevent the server from gaining further information through collusion. In future work, we will explore how to adapt advanced strategies from current malicious secure studies~\cite{lehmkuhl2021muse,chandran2022simc} to further enhance the robustness of PrivTuner.

\textbf{Accelerate FHE with Hardware and GPU Backends.} While our implementation leverages HEXL to accelerate FHE computations, the overhead remains challenging for real-time applications. GPU-based FHE acceleration~\cite{zhang2024secure} has shown promise, with emerging libraries~\cite{wang2023he,yang2024phantom} offering support. However, these are still evolving and seldom used for complex neural networks. As our framework is higher-level and network-oriented, it can seamlessly integrate future GPU-accelerated solutions, enhancing both efficiency and security.

\textbf{Conclusion.}
In summary, this paper introduced a novel scheme named PrivTuner for P3EFT tasks of AI foundation models. Specifically, PrivTuner integrates LoRA with FHE schemes, allowing the model owner server and external mobile devices to collaboratively implement secure and efficient fine-tuning. Additionally, we also introduce a resource allocation method to jointly optimize energy consumption and overall privacy level within PrivTuner. A theoretical analysis of time complexity, solution quality and convergence is provided.
The experimental results demonstrate the effectiveness and superiority of our approach, making it a promising solution to privacy and resource problems in fine-tuning foundation models.

\bibliographystyle{IEEEtran}
\bibliography{ref}

\appendices

\appendices

\section{Approximating Non-linear Functions}\label{sec:Approximations}

Here, we briefly present the 
method of Rovida \textit{et al.}~\cite{rovida2024transformer} to approximate the non-linear functions in BERT-Tiny.\\
\textbf{Softmax.} Given an input vector $\boldsymbol{x} \in \mathbb{R}^d$, the Softmax function can be computed as:
    \begin{align}
        \mathrm{Softmax}(\boldsymbol{x})=\left[\frac{\exp({x_i})}{\sum_{j=1}^{d}\exp({x_j})}\biggm|_{i\in[d]}\right].\label{equa:softmax}
    \end{align}
    The main challenge is to efficiently calculate the underlying exponential function $\mathrm{exp}(x)$ and division computation $1/x$. The exponentiation is approximated using the Maclaurin series:
    \begin{align}
        \mathrm{exp}(x)\approx\sum_{i=0}^n\frac{x^i}{i!}.
    \end{align}
    By doing so, we can limit the average error within $1\times10^{-7}$ with the degree $n=6$, as shown in Table~\ref{Tab:non-linear-layers}.
    
For the division operation $1/x$, we utilize a Chebyshev polynomial supported by OpenFHE to approximate it as follows:
\begin{align}
    1/x=\frac{c_0}{2}+\sum_{i=1}^{n}c_iT_i(x),
\end{align}
where $n$ is the polynomial degree, $T_i(x)=\cos{(i \arccos{x})}$ provide an orthogonal basis of polynomials on the interval $[-1,1]$ with the weight function $1 / \sqrt{1-x^2}$, and $c_i$ is the coefficient for fitting.
Following~\cite{rovida2024transformer}, we set $n=119$ and the input interval for fitting as $[2,5000]$. By doing so, we can observe that the average error can be limited within $1\times10^{-4}$ with a multiplicative depth of $8$ in Table~\ref{Tab:non-linear-layers}.\\
\textbf{GeLU.} The GeLU (Gaussian Error Linear Unit) activation function is defined as follows:
\begin{align}
    \mathrm{GeLU}(x)=\frac{x}{2}\big(1+\mathrm{erf}(\frac{x}{\sqrt{2}})\big),\label{equa:GeLU}
\end{align}
where $\mathrm{erf}(\cdot)$ denotes the Gaussian error function which is given by  $\mathrm{erf}(x)=\frac{2}{\sqrt{\pi}}\int_0^xe^{-t^2}dt$. 
In particular, we approximate the whole GELU function utilizing a Chebyshev polynomial:
\begin{align}
    \mathrm{GeLU}(x)=\frac{c_0}{2}+\sum_{i=1}^{n}c_iT_i(x),
\end{align}
where we set $n=59$, and the average error on interval of $[-18,8]$ could be limited within $4\times10^{-8}$, as shown in Table~\ref{Tab:non-linear-layers}.\\
\textbf{LayerNorm.} 
For a given vector $\boldsymbol{x}\in \mathbb{R}^d$, the LayerNorm function is defined as follows:
\begin{align}
    \mathrm{LayerNorm}(\boldsymbol{x}) = \Big[\frac{(x_i-\mu)}{\sigma}\cdot \gamma +\beta \mid_{i=1,\ldots,d}\Big],
\end{align}
where $\mu=\sum_{i=1}^dx_i/d$ and $\sigma=\sqrt{\sum_{i=1}^d(x_i-\mu)^2}$ are mean and standard deviation, and $\gamma$ and $\beta$ are affine transform parameters. A precomputed LayerNorm is implemented, where the values of $\mu$ and $\sigma$ are experimentally observed and precomputed. Hence, the LayerNorm function could be simplified as:
\begin{align}
    \mathrm{LayerNorm}(\boldsymbol{x}) \approx \Big[(x_i-\hat{\mu})\cdot\hat{\sigma}\cdot\gamma+\beta \mid_{i=1,\ldots,d}\Big],
\end{align}
where $\hat{\mu}$ and $\hat{\sigma}$ are the precomputed mean vectors for $\mu$ and $1/\sigma$. By doing so, we could simplify a lot of circuits at the expense of a 0.022 accuracy drop on SST-2, as reported by~\cite{rovida2024transformer}. For better accuracy performance, we can only precompute the value of $\sigma$, as the computation of $\mu=\sum_{i=1}^dx_i/d$ is actually linear and easily supported by FHE.\\
\textbf{Pooler.} After stacked encoders, we implement a pooling layer as the task layer. A fully connected layer is evaluated on the vector of [CLS] token, followed by an activation function of the hyperbolic tangent: $\mathrm{tanh}(x)=\mathrm{sinh}(x)/\mathrm{cosh}(x)$. Similarly, we view $\mathrm{tanh}(x)$ as a whole and use a Chebyshev polynomial to do the approximation. The observed interval of approximation required is $[-20,20]$, and a $119$-degree polynomial is enough to limit the average error within $3\times10^{-5}$.

\begin{table}[h]
  \caption{Non-linear Functions Approximations.}
  \label{Tab:non-linear-layers}
  \centering
  \begin{tabular}{l|p{2.6cm}p{0.95cm}p{0.6cm}l}
    \hline
      Function &Method & Interval & Degree  & Avg Error \\
    \hline
    exp(x) &Maclaurin series &[-1,1] &6  & $3.94\times10^{-8}$ \\\
     1/x &Chebyshev polynomial &[2,5000] &200  &$9.18\times10^{-5}$ \\
    GeLU(x) &Chebyshev polynomial  &[-18,8] &59  & $3.99\times10^{-8}$ \\
    tanh(x) &Chebyshev polynomial  &[-20,20] &119 & $2.51\times10^{-5}$ \\\hline
\end{tabular}
\end{table}

\section{Count of arithmetic operations for simulation}\label{Appendix:count}

We provide a rough estimation of the count of FHE operations based on the employed FHE-BERT-Tiny model~\cite{rovida2024transformer} and its utilized FHE techniques.

Specifically, we first report the count of operations to the basic supporting algorithms~\cite{rovida2024transformer} (e.g., RotSum) in Table~\ref{Tab:count_operations}. This could be done by checking their released codes at \url{https://github.com/narger-ef/FHE-BERT-Tiny/tree/main}. In our experiments, we found the major bulk of linear computation overhead is caused by calling these algorithms for large matrix-matrix multiplications. Thus, we are able to roughly estimate the number of different operations by counting the number of calls to these algorithms. For instance, we could set the number of tokens as $10$, and the numbers of multiplications, additions and rotations in a BERT-Tiny encoder can be calculated as shown in Table~\ref{tab:count_operations_encoder}. Hence, the number of multiplications, additions and rotations could be set into $260$, $1660$ and $1460$ for a BERT-Tiny model consisting of $2$ encoders. Since the number of operations is considered constant in the optimization problem, it does not affect the solution process.

\begin{table}[h]
  \caption{The count of operations to basic algorithms, and $L$ denotes the number of tokens for an input.}
  \label{Tab:count_operations}
  \centering
  \begin{tabular}{l|ccc}
    \hline
       Algorithm & \#Multiplications &\#Additions & \#Rotations\\\hline
    RotSum &$0$ &$7$ & $7$\\
    MatMulRE & $L$ &$8L$ & $7L$ \\
    MatMulCR  & $L$ &$8L$ & $7L$ \\
    MatMulRELarge  & $4L$ &$32L$ & $31L$ \\
    MatMulCRLarge  & $4L$ &$11L$ & $7L$ \\\hline
\end{tabular}
\end{table}

\begin{table}[h]
  \caption{The count of operations in an encoder for a sentence with $10$ tokens.}
  \label{tab:count_operations_encoder}
  \centering
  \begin{tabular}{l|ccc}
    \hline
       Layer & \#Multiplications &\#Additions & \#Rotations\\\hline
    Self-Attention &$50$ &$400$ & $350$ \\
    FFNN &  $80$ &$430$ & $380$ \\\hline
    Total &$130$  &$830$ &$730$ \\\hline
\end{tabular}
\end{table}

\section{More details of Algorithm \ref{Algorithm:BB}}\label{Appendix:BB}


\textbf{Problem Relaxation.} The discrete variable $\bm{\lambda}$ is relaxed into a continuous and reformulate a relaxed problem 
\begin{align}
\mathbb{P}_3:~\min_{\bm{\lambda}}~& (\ref{Subproblem1:v1})
 \nonumber\\
\text{s.t.}~ & 
(\ref{Constra:lambda_1}),(\ref{Constra:lambda_2}).\nonumber\\
& \lambda_{\textnormal{o}1} \leq \lambda_n \leq \lambda_{\textnormal{o}m},~\forall n \in \mathcal{N}\label{constra:lambda_con},
\end{align}
where $(\ref{constra:lambda_con})$ is a convex constraint now, and thus the overall convexity of problem $\mathbb{P}_3$ can also be guaranteed.\\
\textbf{Solution of the Relaxed Problem.} 
With convexity, $\mathbb{P}_3$ can be solved by obtaining the Karush-Kuhn-Tucker (KKT) point. The partial Lagrange function of $\mathbb{P}_3$ is given by:
\begin{align}\label{Lagrange_function1}
\mathcal{L}_1(\bm{\lambda},\alpha)\! =\!&
\sum_{n=1}^{N}\! \kappa y_1(\lambda_n)s_n \overline{f}_n^2(\lambda_n)
\!+\! \sum_{n=1}^{N} \kappa y_2(\lambda_n)D_n\overline{g}_n^2(\lambda_n) \nonumber\\
&- \omega \sum_{n=1}^{N}\sigma_ny_6(\lambda_n)
+\alpha(\sum_{n=1}^N \overline{f}_n(\lambda_n) - f^{\textnormal{total}}).
\end{align}
Here we first give the KKT conditions of (\ref{Lagrange_function1}):
\begin{align}
    \frac{\partial \mathcal{L}_1}{\partial \lambda_n} =&
    \frac{6\kappa M_n^3C_0^3(\lambda_n+C_1)^5}{(T^{\textnormal{max}}_{\textnormal{en}})^2}
    \!-\! \omega c_n C_7 \lambda_n \!+\! \alpha \frac{2C_0(\lambda_n\!+\!C_1)M_{n}}{T^{\textnormal{max}}_{\textnormal{en}}} \nonumber\\
    &
    + \frac{3\kappa D_n^3(C_2a_{n}\!+\!C_4m_{n}\!+C_6 o_n\!)^3\lambda_n^2}{(T^{\textnormal{max}}_{\textnormal{ft}})^2}
    \!=\! 0,~\forall n \in \mathcal{N}. \label{Lagrange:partial_lambda}
\end{align}
\begin{align}
    \alpha \cdot (\sum_{n=1}^N \overline{f}_n(\lambda_n) - f^{\textnormal{total}}) =0.\label{Lagrange:dual}
\end{align}
\begin{align}
    \sum_{n=1}^N \overline{f}_n(\lambda_n) \leq f^{\textnormal{total}}.
\end{align}
\begin{align}
    \alpha \geq 0.
\end{align}
From (\ref{Lagrange:partial_lambda}), we derive its solution $\lambda_n$ represented by $\alpha$ as ${\lambda}_n(\alpha)$. Considering (\ref{Constra:lambda_2}) and (\ref{constra:lambda_con}), we further refine the solution:
\begin{align}
   \overline{\lambda}_n(\alpha) = \max\{\min\{\lambda_{om},\lambda_n^{\textnormal{max}},\lambda_n(\alpha)\},\lambda_{o1}\},\forall n \in \mathcal{N}. \label{hatlambda}
\end{align}
Then, we obtain $\overline{\lambda}_n(\alpha=0)$ by assuming $\alpha=0$ in (\ref{hatlambda}) and discuss the following cases to further determine $\alpha$:
\begin{enumerate}
    \item $\sum_{n=1}^N\overline{f}_n(\overline{\lambda}_n(\alpha=0))\leq f^{\textnormal{total}}$. In this case, we can simply set $\alpha=0$, and the solution is $\overline{\lambda}_n(\alpha=0)$.
    \item   $\sum_{n=1}^N\overline{f}_n(\overline{\lambda}_n(\alpha=0)) > f^{\textnormal{total}}$. Setting $\alpha = 0$ will volate  (\ref{Constra:lambda_1}). Substitute $\alpha > 0$ in (\ref{Lagrange:dual}): 
    \begin{align}
    \hspace{-5pt}\sum_{n=1}^N\overline{f}_n(\max\{\min\{\lambda_{\textnormal{o}m},\!\lambda_n^{\textnormal{max}},\!\overline{\lambda}_n(\alpha)\},\!\lambda_{o1}\})\!=\! f^{\textnormal{total}}.\label{bisection:lambda}
    \end{align}
    A bisection method could find a solution $\alpha^*$ to (\ref{bisection:lambda}). Besides, since  $\sum_{n=1}^N\overline{f}_n(\overline{\lambda}_n(\alpha=0) > f^{\textnormal{total}}$, we have:
    \begin{align}
    \overline{\lambda}_n(\alpha=0) > \overline{\lambda}_n(\alpha^{*}).
    \end{align}
\end{enumerate}
Summarize above cases, and a solution $\widehat{\lambda}_n$ is obtained: 
\begin{align}
    \widehat{\lambda}_n = \min \{\overline{\lambda}_n(\alpha=0), \overline{\lambda}_n(\alpha^{*})\},~\forall n \in \mathcal{N}.\label{suboptimal:lambda}
\end{align}\\
\textbf{Branch.} 
If a solution doesn't meet the integrality constraint yet achieves a better objective function value, we branch the solution space into two subproblems. We use problem $\mathbb{P}_3$ and its solution $\widehat{\bm{\lambda}}$ as an example to demonstrate this process:
\begin{itemize}
    \item 
    Select a $\widehat{\lambda}_k$ such that it is not an integer in $\{\lambda_{\textnormal{o}1}, \ldots,\lambda_{\textnormal{o}m}\}$ for separation:  $\lambda_k^{-} = \lfloor \widehat{\lambda}_k \rfloor$ and $
    \lambda_k^{+} = \lceil \widehat{\lambda}_k \rceil,$ where $\lfloor \cdot \rfloor$ and $\lceil \cdot \rceil$ denote floor and ceil functions to integers in $\{\lambda_{\textnormal{o}1}, \ldots,\lambda_{\textnormal{o}m}\}$, and $\lambda_k^{-},\lambda_k^{+}$ are rounded results.
    \item Based on $\mathbb{P}_3$, generate two new subproblems with two additional constraints, seperately:
    \begin{align}
        \textbf{(SUB1):}~ \min_{\bm{\lambda}}~& \text{(\ref{Subproblem1:v1})}\nonumber\\
    \text{s.t.}~ & 
    \text{(\ref{Constra:lambda_1})}, \text{(\ref{Constra:lambda_2})}, \text{(\ref{constra:lambda_con})},\nonumber\\
    & \lambda_{o1} \leq \lambda_k \leq \lambda_k^{-}. \nonumber \\
    \textbf{(SUB2):}~\min_{\bm{\lambda}}~& \text{(\ref{Subproblem1:v1})}\nonumber\\
    \text{s.t.}~ & \text{(\ref{Constra:lambda_1}), (\ref{Constra:lambda_2}), (\ref{constra:lambda_con})},\nonumber\\
    & \lambda_k^{+} \leq \lambda_k \leq \lambda_{om}.  \nonumber
    \end{align}
\end{itemize}


 \setlength{\abovedisplayskip}{2.5pt plus 1pt minus 0pt}
\setlength{\belowdisplayskip}{2.5pt plus 1pt minus 0pt}
\setlength\abovedisplayshortskip{1pt plus 1pt minus 0pt}
\setlength\belowdisplayshortskip{1pt plus 1pt minus 0pt}

\section{Proof of Theorem \ref{theorem:optimal_p_B}}\label{Appendix:p,B}

We first write down the partial Lagrange function of $\mathbb{P}_4$:
\begin{align}
\mathcal{L}_2(\bm{p},\bm{B},\beta,\bm{\gamma})&=\sum_{n=1}^N\big((p_nd_{n})^2z_{n}+ \frac{1}{4(r_n)^2z_{n}}\big)
\nonumber\\
&\hspace{-10pt}+\beta(\sum_{n=1}^N B_{n}-B^{\textnormal{total}}) +\sum_{n=1}^N\gamma_n(r_n^{\textnormal{min}}-r_n).
\end{align}
The corresponding KKT conditions are as follows:\\
\textbf{Stationarity:}
\begin{subequations}
\begin{align}
    \frac{\partial \mathcal{L}_2}{\partial p_n} = &2p_nd_{n}^2z_{n}-(\frac{1}{2r_n^3z_{n}}\!+\!\gamma_n)\frac{\partial r_n}{\partial p_n}=0,~\forall n \in \mathcal{N},\label{Lagrange:partial_p}\\
   \frac{\partial \mathcal{L}_2}{\partial B_n} =& -\!(\frac{1}{2r_n^3z_n}+\gamma_n)\frac{\partial r_n}{\partial B_n}+\beta=0,~\forall n \in \mathcal{N}. \label{Lagrange:partial_B}
\end{align}
\end{subequations}\\
\textbf{Comlementary Slackness:}
\begin{subequations}
\begin{align}
    &\beta\cdot(\sum_{n=1}^N B_n- B^{\textnormal{total}}) = 0, \label{Lagrange:comp_beta}\\
    &\gamma_n\cdot(r_{n}^{\textnormal{min}}- r_n) = 0,~\forall n \in \mathcal{N}.\label{Lagrange:comp_gamma}
\end{align}
\end{subequations}\\
\textbf{Primal feasibility:} 
(\ref{Constra:p}), (\ref{Constra:B}).\\
\textbf{Dual feasibility:}
\begin{subequations}\label{Dualfeasibility}
\begin{align}
    &\text{(\ref{Dualfeasibility}a):~}\beta \geq 0,~~~
    \text{(\ref{Dualfeasibility}b):~}\gamma_n \geq 0,~\forall n \in \mathcal{N}.\notag
\end{align}
\end{subequations}\\
From (\ref{Lagrange:partial_p}), we could derive its solution $p_n$ as a function of $B_n$ and $\gamma_n$: $p_n(B_n,\gamma_n)$. Considering constraint (\ref{Constra:p}), we further refine $p_n(B_n,\gamma_n)$ as follows:
\begin{align}
    \overline{p}_n(B_n,\gamma_n)=\min\{p_n^{\textnormal{max}}, p_n(B_n,\gamma_n)\}.\label{temp:1}
\end{align}\\
Similarly, substituting $\overline{p}_n(B_n,\gamma_n)$ in (\ref{Lagrange:partial_B}), we can derive its solution $\overline{B}_n(\beta,\gamma_n)$.
Next, if $\gamma_n=0$ and given a feasible $\beta$, we have $\hat{B}_n=\overline{B}_n(\beta,\gamma_n=0)$ and $\hat{p}_n=\overline{p}_n(\hat{B}_n,\gamma_n=0)$. Thus, the corresponding transmission rate is $\hat{r}_n=r_n(\hat{p}_n,\hat{B}_n)$. An analysis of the following two cases based on $\hat{r}_n$ is provided:
\begin{enumerate}
    \item $r_n^{\textnormal{min}} \leq \hat{r}_n$. In this case, we can simply set $\gamma_n=0$, and all other KKT conditions are satisfied.
    \item $r_n^{\textnormal{min}} > \hat{r}_n$. Here, we set $\gamma_n>0$ and make sure $r_n = r_n^{\textnormal{min}}$ according to (\ref{Lagrange:comp_gamma}). Given this condition, and by utilizing $\overline{B}_n(\beta,\gamma_n)$ alongside $\overline{p}_n(\overline{B}_n,\gamma_n)$, we could derive the solution $\gamma_n$ denoted as a function of employed $\beta:\gamma_n(\beta)$.
\end{enumerate}
Summarize both cases, and we have:
\begin{align}\label{optimal:gamma}
    \overline{\gamma}_n(\beta) = 
    \begin{cases}
    0, &\text{if } r_n^{\textnormal{min}} \leq r_n(\hat{p}_n,\hat{B}_n) \\
    \gamma_n(\beta), &\text{otherwise}.
    \end{cases}
\end{align}
Next, we further refine the value of multiplier $\beta$. 
It is obvious from (\ref
{Lagrange:partial_B}) that $\beta>0$, thus from (\ref{Dualfeasibility}a) we have:
\begin{align}
    \sum_{n=1}^N\overline{B}_n\big(\beta>0,\overline{\gamma}_n(\beta)\big)-B^{\textnormal{total}}=0.\label{temp:2}
\end{align}
A bisection method could solve (\ref{temp:2}) and derive the  solution $\beta^*$. Subsitute $\beta^*$ in (\ref{optimal:gamma}) and we have $\gamma_n^*=\overline{\gamma}_n(\beta^*)$.
With $\beta^*$ and $\gamma_n^*$, the optimal $p_n^*$ and $B_n^*$ can be expressed as follows:
\begin{align}
    B_n^* &= \overline{B}_n(\beta^*,\gamma_n^*),\label{equa:optimal_B}\\
    p_n^* &= \overline{p}_n(B_n^*,\gamma_n^*),\label{equa:optimal_p}
\end{align}
where functions $\overline{B}_n(\cdot,\cdot)$ and $\overline{p}_n(\cdot,\cdot)$ have been defined in the above analysis via (\ref{temp:1}) and the sentence succeeding it. \qed

\ifCLASSOPTIONcaptionsoff
  \newpage
\fi

\end{document}